\documentclass[11pt]{article}
\usepackage{cite}      
\usepackage{graphicx}  
\usepackage{caption}
\usepackage{times}
    
\usepackage{amssymb}
\usepackage{mathrsfs}
\usepackage{amsmath}   
\usepackage{amsthm}

\usepackage[sans]{dsfont}
\usepackage{stmaryrd}
\usepackage{pifont}
\usepackage{wasysym}
\usepackage{algorithm} 
\usepackage{array}
\usepackage{url}
\usepackage{tabu}
\usepackage{multirow}
\usepackage{booktabs}

\usepackage{algpseudocode}
\usepackage{algorithm}

\usepackage{authblk}

\newtheorem{theorem}{Theorem}[section]

\newtheorem{lemma}[theorem]{Lemma}

\newtheorem{definition}[theorem]{Definition}

\newcommand{\Pj}{ { {\tt Pj}_{\sigma^T_i} } }

\def\CC{\mathbb C}
\def\RR{\mathbb R}
\def\ZZ{\mathbb Z}

\def\cA{\mathcal A}
\def\cB{\mathcal B}

\def\cD{\mathcal D}

\def\cI{\mathcal I}

\def\cL{\mathcal L}

\def\cN{\mathcal N}
\def\cS{\mathcal S}

\def\cX{\mathcal X}

\def\bd{\mathbf d}

\def\bzero{\mathbf 0}

\newcommand{\hide}[1]{}
\newcommand{\raf}[1]{(\ref{#1})}

\newcommand{\cP}{\ensuremath{\mathcal{P}}}

\newcommand{\cR}{\ensuremath{\mathcal{R}}}
\newcommand{\cV}{\ensuremath{\mathcal{V}}}

\newcommand{\Z}{\ensuremath{\mathbb Z}}

\newcommand{\T}{\ensuremath{\mathcal T}}

\newcommand{\OPT}{\ensuremath{\textsc{Opt}}}

\newcommand{\argmin}{\ensuremath{\mathrm{argmin}}}

\newcommand{\poly}{\operatorname{poly}}

\newcommand{\argmax}{\operatorname{argmax}}

\title{Truthful Mechanisms for Combinatorial Allocation of Electric Power in Alternating Current Electric Systems for Smart Grid\footnote{This paper appears in ACM Transactions on Economics and Computation, Vol. 5, No. 1, Article 7, October, 2016. DOI: http://dx.doi.org/10.1145/2955089 Extended abstracts containing some of the results have been presented in the International Conference on
Autonomous Agents and Multi-Agent Systems (AAMAS 2014) \cite{CEK14CKS} and a workshop at the International Conference on Computer Communication and Networks (ICCCN 2014)\cite{KCE14}.}}

\author{Chi-Kin Chau$^\ast$}
\author{Khaled Elbassioni$^\ast$}
\author{Majid Khonji$^\dag$}
\affil{$^\ast$Department of EECS \\ Masdar Institute of Science and Technology, Abu Dhabi, UAE}
\affil{$^\dag$Department of Research and Development \\ Dubai Electricity and Water Authority (DEWA), Dubai, UAE} 

\affil{Email: \{ckchau, kelbassioni\}@masdar.ac.ae, majid.khonji@dewa.gov.ae}
\begin{document}

\maketitle

\begin{abstract}
	Traditional studies of combinatorial auctions often only consider linear constraints. The rise of smart grid presents a new class of auctions, characterized by quadratic constraints. This paper studies the {\em complex-demand knapsack problem}, in which the demands are complex valued and the capacity of supplies is described by the magnitude of total complex-valued demand. This naturally captures the power constraints in alternating current (AC) electric systems. In this paper, we provide a more complete study and generalize the problem to the multi-minded version, beyond the previously known $\frac{1}{2}$-approximation algorithm for only a subclass of the problem. More precisely, we give a truthful PTAS for the case $\phi\in[0,\frac{\pi}{2}-\delta]$, and a truthful FPTAS, which {\it fully} optimizes the objective function but violates the capacity constraint by at most $(1+\epsilon)$, for the case $\phi\in(\frac{\pi}{2},\pi-\delta]$, where $\phi$ is the maximum argument of any complex-valued demand and $\epsilon,\delta>0$ are arbitrarily small constants. We complement these results by showing that, unless P=NP, neither a PTAS for the case $\phi\in(\frac{\pi}{2},\pi-\delta]$ nor any bi-criteria approximation algorithm with polynomial guarantees for the case when $\phi$ is arbitrarily close to $\pi$ (that is, when $\delta$ is arbitrarily close to $0$) can exist.
\end{abstract}


\section{Introduction} \label{sec:intro} 

Traditionally, many practical auction problems are combinatorial in nature, requiring carefully designed time-efficient approximation algorithms. Although there have been decades of research in approximating combinatorial auction problems, traditional studies of combinatorial auctions often only consider linear constraints. Namely, the demands for certain goods are limited by the respective supplies, described by certain linear constraints.

Recently, the rise of smart grid presents a new class of auction problems. In alternating current (AC) electric systems \cite{GS94power}, the power is determined by time-varying voltage and current, which gives rise to two types of power demands: (1) {\em active} power (that can be consumed by resistors at the loads) and, (2) {\em reactive} power (that continuously bounces back and forth between the power sources and loads). The combination of active power and reactive power is known as {\em apparent} power. The ratio between active power and apparent power is known as {\em power factor}. In practice, the physical capacity of power generation and transmission is often expressed by apparent power. Electric appliances and instruments with capacitive or inductive components have non-zero reactive power. However, most electric appliances and instruments are subject to regulations to limit their power factors \cite{NEC}. It is vital to ensure that the total power usage is within the apparent power constraint, given the maximum power factor of power demands.

In the common literature of electric power systems \cite{GS94power}, apparent power is represented by a complex number, wherein the real part represents the active power and the imaginary part represents the reactive power. Hence, it is often necessary to use a quadratic constraint, namely the magnitude of complex numbers, to describe the system capacity. The power factor is related to the phase angle between active power and reactive power. Yu and Chau \cite{YC13CKS} introduced the {\em complex-demand knapsack problem} (CKP) to model a one-shot auction for combinatorial AC electric power allocation, which is a quadratic programming variant of the classical knapsack problem.

Furthermore, future smart grids will be automated by agents representing individual users. Hence, one might expect these agents to be self-interested and may untruthfully report their valuations or demands. This motivates us to consider truthful (aka. incentive-compatible) approximation mechanisms, in which it is in the best interest of the agents to report their true parameters. In \cite{YC13CKS} a monotone $\frac{1}{2}$-approximation algorithm that induces a deterministic truthful mechanism was devised for the complex-demand knapsack problem, which, however, assumes that all complex-valued demands lie in the positive quadrant. 

In this paper, we provide a complete study and generalize the complex-demand knapsack problem to the multi-minded version, beyond the previously known $\frac{1}{2}$-approximation algorithm. More precisely, we consider the problem under the framework of (bi-criteria) $(\alpha,\beta)$-approximation algorithms, which compute a feasible solution with objective function within a factor of $\alpha$ of optimal, but may violate the capacity constraint by a factor of at most $\beta$. 
We give a (deterministic) truthful $(1-\epsilon,1)$-approximation algorithm for the case $\phi\in[0,\frac{\pi}{2}-\delta]$, and a truthful $(1,1+\epsilon)$-approximation for the case $\phi\in(\frac{\pi}{2},\pi-\delta]$, where $\phi$ is the maximum argument of any complex-valued demand   and $\epsilon,\delta>0$ are arbitrarily small constants. 
Moreover, the running time in the latter case is polynomial in $n$ and $\frac{1}{\epsilon}$ (this may be thought of as an {\it FPTAS with resource augmentation}; see, e.g., \cite{KP2000,PSTW02}). We complement these results by showing that, unless P=NP, neither a PTAS can exist for the latter case nor any bi-criteria approximation algorithm with exponential guarantees for the case when $\phi$ is arbitrarily close to $\pi$. Note that the difficulty when $\phi\in(\frac{\pi}{2},\pi]$ is mainly due to the fact that demands are allowed to have both positive and negative real parts, which can cancel each other; this allows an optimal solution to pack much larger set of demands, within the available capacity, than any polynomial time algorithm can detect. We remark also that \cite{woeginger2000does,YC13CKS} show no FPTAS exists for the case $\phi\in[0,\frac{\pi}{2}-\delta]$. Therefore, our results completely settle the open questions in \cite{YC13CKS}.

\subsection{Contribution}
In Table \ref{tab}, we briefly list the inapproximability and the best known truthful mechanisms for the {\em $m$-dimensional knapsack problem} ({\sc $m$DKP}), $m\ge2$, along with
 our results for three classes of {\sc CKP} (namely, demands with maximum argument $\phi \in [0,\frac{\pi}{2}]$, $\phi \in [0, \pi - \delta]$, and $\phi \ge \pi - \delta'$ where $\delta$ is {\em polynomially} small in $n$, while $\delta'$ is {\em exponentially} small $n$).

\begin{table}[h!]  \hspace{-10pt}

\centering{ \footnotesize

\begin{tabular} {l c c c c}
 \toprule
 & {\sc  CKP$[0,\frac{\pi}{2}]$} &{\sc CKP$[\frac{\pi}{2}+\delta,\pi \mbox{-}\delta]$}  & {\sc CKP$[\pi\mbox{-}\delta',\pi]$} & {\sc $m$DKP}\\
\midrule[1pt]
  \parbox{3cm}{\bf Inapproximability}  & \parbox{2.5cm}{\centering No FPTAS \cite{YC13CKS,woeginger2000does}} & \parbox{2cm}{\centering No $(\alpha, 1)$-approx (Sec. \ref{sec:hardness}) }
  & {\parbox{2cm}{\centering
  Bi-criteria Inapproximable (Sec. \ref{sec:hardness}) }} & \parbox{3cm}{\centering No FPTAS \\  \centering(see, e.g.,  \cite{KPP10book})}\\
 \cmidrule{1-5}
\parbox{3cm}{\bf Truthful Mechanism} &   \parbox{2.5cm}{\centering PTAS (Sec. \ref{sec:truthful-ptas})\\
\centering $(1,1+\epsilon)$-FPTAS (Sec. \ref{sec:tbp})	}&  \parbox{2.5cm}{\centering $(1,1+\epsilon)$-FPTAS (Sec. \ref{sec:tbp}) }  & None & \parbox{2.5cm}{\centering PTAS \cite{DN10}\\ \centering Bi-criteria FPTAS \cite{KTV13}}\\
\bottomrule
\end{tabular}
}
\caption{A summary of results}
\label{tab}
\end{table}

\section{Related Work} \label{sec:related} 

Linear combinatorial auctions can be formulated as variants of the classical knapsack problem \cite{CK00,KPP10book,FC84alg}. Notably, these include the {\em one-dimensional knapsack problem} ({\sc 1DKP}) where a single item has multiple copies, and its multi-dimensional generalization, the {\em $m$-dimensional knapsack problem} ({\sc $m$DKP}).
There is an FPTAS for {\sc 1DKP} (see, e.g.,  \cite{KPP10book}). 

In mechanism design setting, where each customer may untruthfully report her valuation and demand, it is desirable to design {\it truthful} or  {\it incentive-compatible} approximation mechanisms, in which it is in the best interest of each customer to reveal her true valuation and demand \cite{DN07}.
In the so-called {\it single-minded case}, a {\it monotone} procedure can guarantee incentive compatibility \cite{NRTV07}.  While the straightforward FPTAS for  {\sc 1DKP} is not monotone, since the scaling factor involves the maximum item value, \cite{BKV05KS} gave a monotone FPTAS, by performing the same procedure with a series of different scaling factors irrelevant to the item values and taking the best solution out of them. Hence, {\sc 1DKP} admits a truthful FPTAS. We remark that monotonicity  may be not enough for the incentive compatibility in the general setting. More recently, a truthful PTAS, based on another approach using dynamic programming and the notion of the  {\it maximal-in-range} mechanism, was given in \cite{DN10} for the {\it multi-minded} case. We will use the maximal-in-range approach in this paper.     

As for {\sc $m$DKP} with $m\geq 2$, a PTAS is given in \cite{FC84alg} based on the integer programming formulation, but it is not evident to see whether it is monotone.     
On the other hand, {\sc 2DKP} is already inapproximable by an FPTAS unless P = NP, by a reduction from {\sc equipartition} \cite{KPP10book}.  Very recently, \cite{KTV13} gave a truthful FPTAS with $(1+\epsilon)$-violation for multi-unit combinatorial auctions with a constant number of distinct goods (including {\sc $m$DKP}), and its generalization to the multi-minded version, when $m$ is fixed. Their technique is based on applying the VCG-mechanism to a rounded problem. Based on the PTAS for the {\it $m$-minded multi-unit auctions}
 developed in \cite{DN10}, they also obtained a truthful PTAS for $m$-minded multi-unit combinatorial auctions with a constant number of distinct goods.  
 {Intuitively, a valuation function is $m$-minded if it is completely determined by the values on $m$ different choices; for simplicity we call this type of valuation {\it multi-minded} in the rest of the paper.}

In contrast, truthful non-linear combinatorial auctions were explored to a little extent.  Yu and Chau  \cite{YC13CKS} introduced the complex-demand knapsack problem, which models auctions with a convex quadratic constraint. An earlier paper \cite{woeginger2000does} also introduced the same problem without considering truthfulness by a different name called $2$-weighted knapsack problem. One can regard, the complex-demand knapsack problem with strategic considerations as an auction design problem, where users bid on complex-valued items, and a feasible solution allocates one item to each user such that the total magnitude of allocated items is below a certain threshold. Even though some of the existing techniques can deal with combinatorial auctions with convex non-linear relaxations (see, e.g., \cite{LS11}), those techniques require bounded integrality gap and yield randomized truthful-in-expectations mechanisms. 

\section{Problem Definitions and Notations}\label{sec:model}
In this section we formally define the complex-demand knapsack problem. We present first the non-strategic version of the problem where we assume all parameters are known beforehand. Then we describe the strategic version where each user $k$ declares his/her valuation function defined over a set of declared demands. In the latter setting, we consider the case where users could lie about their valuation functions and demand sets in order to optimize their utility functions (see Sec.~\ref{sec:model.truth} for a formal definition of utility). 
Towards the end of this section, we present an application of the complex-demand knapsck problem to power allocation in  (AC) alternating current electric systems.

\subsection{Complex-demand Knapsack Problem (non-strategic version) } 

We adopt the notations from \cite{YC13CKS}. 
Our study concerns power allocation under a capacity constraint on the magnitude of the total satisfiable demand (i.e., apparent power). Throughout this paper, we sometimes denote $\nu^{\rm R} \triangleq {\rm Re}(\nu)$ as the real part and $\nu^{\rm I} \triangleq {\rm Im}(\nu)$ as the imaginary part of a given complex number $\nu$. We also interchangeably denote a complex number by a 2D-vector as well as a point in the complex plane. $|\nu|$ denotes the magnitude of $\nu$.

We define the non-strategic version of the complex-demand knapsack problem ({\sc CKP}) with a set $[n]\triangleq\{1,\ldots,n\}$ of users as follows: 
\begin{eqnarray}
\textsc{(CKP)} \qquad& \displaystyle \max_{x \in \{0, 1 \}^n} \sum_{k\in[n]}v_k x_k \label{CKP}\\
\text{subject to}\qquad & \displaystyle \Big|\sum_{k\in \cN}d_k x_k\Big| \le C. \label{C1}
\end{eqnarray}
where $d_k = d_k^{\rm R} + {\bf i} d_k^{\rm I} \in\CC$ is the {\em complex-valued} demand of power for the $k$-th user, $C \in\RR_+$ is a real-valued capacity of total satisfiable demand in apparent power, and $v_k \in \RR_+$ is the valuation of user $k$ if her demand $d_k$ is satisfied (i.e., $x_k = 1$). 
Evidently, {\sc CKP} is also NP-complete, because the classical 1-dimensional knapsack problem ({\sc 1DKP}) is a special case. 

We define a class of sub-problems for {\sc CKP}, by restricting the maximum phase angle (i.e., the argument) of any demand. In particular, we will write {\sc CKP}$[\phi_1,\phi_2]$ for the restriction of problem {\sc CKP} subject to $\phi_1 \le \max_{k \in \cN}{\rm arg}(d_k)$ $\le \phi_2$, where ${\rm arg}(d_k)\in [0,\pi]$. We remark that in the realistic settings of power systems, the active power demand is positive (i.e., $d_k^{\rm R} \ge 0$), but the power factor (defined by $\frac{d^{\rm R}_k}{|d_k|}$) is bounded by a certain threshold, which is equivalent to restricting the argument of complex-valued demands. 

From the computational point of view, we will need to specify how the inputs are described. Throughout the paper we will assume that each of the demands is given by their real and imaginary components, represented as rational numbers.

\subsection{Non-single-minded Complex Knapsack Problem (strategic version)}
In this paper, we extend the single-minded {\textsc CKP} to general {\it non-single-minded} version, and then we apply the well-known {\it VCG-mechanism}, or equivalently the framework of {\it maximal-in-range} mechanisms \cite{NR07}. The non-single-minded version is defined as follows. By a slight abuse of notation, we denote $v_k(\cdot)$ as a valuation function for the non-single-minded setting. As above we assume a set $\cN$ of $n$ users: user $k$ has a valuation function $v_k(\cdot):\cD\to\RR_+$ over a (possibly infinite) set of demands $\cD\subseteq\CC$. We assume that $\bzero\in\cD$, $v_k(\bzero)=0$ for all $k\in\cN$, and w.l.o.g., $|d|\le C$ for all $d\in\cD$. We further assume that each $v_k(\cdot)$ is {\it monotone} with respect to a partial order ``$\preceq$'' defined on the elements of $\CC$ as follows: for $d,f\in\CC$, $d\succeq f$ if and only if
{\small
$$
|d^{\rm R}|\ge |f^{\rm R}|,|d^{\rm I}|\ge |f^{\rm I}|,{\rm sgn}(d^{\rm R}) = {\rm sgn}(f^{\rm R}),{\rm sgn}(d^{\rm I}) = {\rm sgn}(f^{\rm I}).
$$
}
(See Fig.~\ref{fig:monotone} for pictorial illustration.)
We assume $\bzero\preceq d$ for all $d\in\cD$. Then for all $k\in\cN$, the monotonicity of $v_k(\cdot)$ means that $v_k(d)\ge v_k(f)$ whenever $d\succeq f$.
\begin{figure}[!ht]
	\begin{center} 
		\includegraphics[scale=0.8]{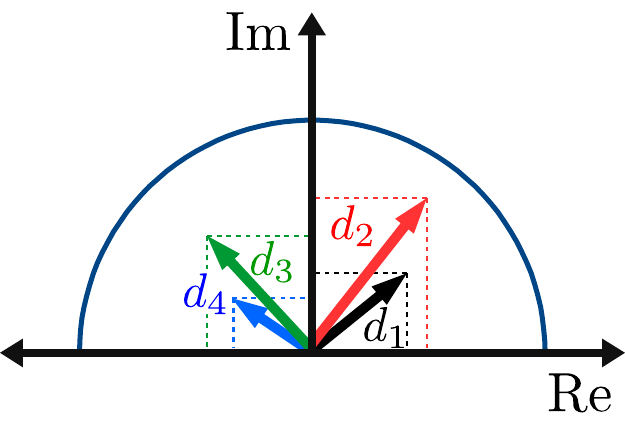}
	\end{center} 
	\caption{A pictorial illustration for the partial order ``$\preceq$'': $d_1\preceq d_2$ and $d_4\preceq d_3$.}
	\label{fig:monotone}
\end{figure} 

The non-single-minded problem can be described by the following program (in the variables $d_k$): 
\begin{eqnarray}
\textsc{(NsmCKP)}&\displaystyle \max  \sum_{k\in\cN}v_k (d_k) \label{(CV-nsm-KS)}\\
\text{s.t.} & \displaystyle \bigg|\sum_{k\in\cN} d_k\bigg|\le C \label{nsm-CV1}\\
& d_k\in\cD \text{ for all }k\in \cN, \label{nsm-CV2}
\end{eqnarray}
where $ |\sum_{k\in\cN} d_k | =\sqrt{(\sum_{k\in\cN} d_k^{\rm R} )^2 + (\sum_{k\in\cN} d_k^{\rm I} )^2}$.
Of particular interest is the {\it multi-minded} version of the problem ({\sc MultiCKP}), defined as follows. Each user $k\in\cN$ is interested only in a {\it polynomial-size} subset of demands $D_k\subseteq \cD$ and declares her valuation only over this set  (that is, 
$D_k$ is a set of 2-dimensional vectors of size $|D_k|=\poly(n)$). Note that the multi-minded problem can be modeled in the form \textsc{(NsmCKP)} by assuming w.l.o.g. that $\bzero\in D_k$, for each user $k\in\cN$, and defining the valuation function $v_k(\cdot):\cD\to\RR_+$ as follows: 
\begin{equation}\label{mm-val}
v_k(d)=\max_{d_k\in D_k}\{v_k(d_k):~d_k\preceq d \}.
\end{equation}
We shall assume that the demand set of each user lies completely in one of the quadrants, namely, either $d^{\rm R}\ge 0$ for all $d\in D_k$, or $d^{\rm R}< 0$ for all $d\in D_k$.  
This assumption is needed  for the results in Sec.~\ref{sec:tbp}, as we will see, the problem is split into two independent {\sc $2$-DKP} problems based on the demands' quadrants, then each user $k$ is allocated by one demand in the quadrant where all her demands $D_k$ lie. (As we will see in Sec.~\ref{sec:app}, $d^{\rm R}\ge 0$ corresponds to an inductive load, while $d^{\rm I}$ corresponds to a capacitive load.)
Note that the single-minded version (which is \textsc{CKP}) is special case, where $|D_k|=1$ for all $k$. 
When we consider strategic users, we will assume that the users can lie about their demand sets and/or valuation functions (as long as the demand set of each user lies completely in one of the quadrants).

We will write {\sc MultiCKP}$[\phi_1,\phi_2]$ for the restriction of the problem  subject to $\phi_1 \le \phi \le \phi_2$ for all $d \in \cD$ where $\phi\triangleq\max_{d\in\cD}{\rm arg}(d)$ (and as before we assume ${\rm arg}(d)\ge 0$). 

\subsection{Non-single-minded Multidimensional Knapsack Problem}\label{MCMDKS-sec}
To design truthful mechanisms for \textsc{NsmCKP}, it will be useful to consider the {\it non-single-minded multidimensional knapsack} problem\footnote{Sometimes, this is also called the {\it multiple-choice knapsack problem}.} (\textsc{Nsm-$m$DKP}) defined as follows, where we assume more generally that $\cD\subseteq\RR_+^m$ and a {\it capacity vector} $c\in\RR_+^m$ is given. As before, a valuation function for each user $k$ is given by \raf{mm-val}. 
An {\it allocation} is given by an assignment of a demand $d_k=(d_k^1,...,d_k^m)\in\cD$ for each user $k$, so as to satisfy the $m$-dimensional capacity constraint $\sum_{k\in\cN}d_k\le c$. The objective is to find an allocation $\bd=(d_1,\ldots,d_n)\in\cD^n$ so as to maximize the sum of the valuations $\sum_{k\in\cN}v_k(d_k)$. The problem can be described by the following program:

\begin{eqnarray}
\textsc{(Nsm-$m$DKP)} & \displaystyle \max\sum_{k\in\cN}v_k (d_k)& \label{(mCmD-KS)}\\
\text{ s.t. } & \displaystyle \sum_{k\in\cN} d_k &\le  c \label{mCmD-1}\\
& d_k\in\cD \text{ for all }k\in \cN. \label{mCmD-3}
\end{eqnarray}

Similarly, we consider the {\it multi-minded} version of the problem ({\sc Multi-$m$DKP}): each user $k\in\cN$ is interested only in a {\it polynomial-size} subset of demands $D_k\subseteq \cD$ and declares her valuation only over this set. The multi-minded problem can be modeled in the form \textsc{Nsm-$m$DKP} by assuming w.l.o.g. that $\bzero\in D_k$, for each user $k\in\cN$, and defining the valuation function $v_k(\cdot):\cD\to\RR_+$ as 
$v_k(d)=\max_{d_k\in D_k}\{v_k(d_k):~d_k\preceq d \}$,
where ``$\preceq$" is a component-wise partial order.

It is worth noting that \textsc{Nsm-$m$DKP} is similar to multi-unit combinatorial auctions (CA) with $m$ distinct goods; the difference is that in the latter problem the set $\cD$ is restricted to be integral, whereas we do not assume this restriction in \textsc{Nsm-$m$DKP}.

\subsection{Approximation Algorithms}

We present an explicit definition of approximation algorithms for our problem.
A feasible allocation satisfying \raf{nsm-CV1} is represented by a vector $\bd=(d_1,\ldots,d_n) \in \cD^n$. When no demand is allocated to user $k$, we assume $d_k=\bzero$. We write $v(\bd)\triangleq\sum_{k\in\cN}v_k(d_k)$.
Let $ \bd^\ast$ be an optimal allocation of \textsc{NsmCKP} (or \textsc{MultiCKP}) and $\OPT \triangleq v(\bd^\ast)$ be the corresponding total valuation. We are interested in polynomial time algorithms that output an allocation that is within a factor $\alpha$ of the optimum total valuation, but may violate the capacity constraint by at most a factor of $\beta$:  
\begin{definition}
For $\alpha\in(0,1]$ and $\beta\ge 1$, a bi-criteria $(\alpha,\beta)$-approximation to \textsc{NsmCKP} is an allocation $(d_k)_{k} \in \cD^n$ satisfying 
\begin{eqnarray}
& \displaystyle \Big|\sum_{k\in\cN}d_k\Big| \le \beta \cdot C \label{C1'}\\
\text{such that}\qquad & \displaystyle \sum_{k\in\cN}v_k(d_k) \ge \alpha \cdot \OPT.
\end{eqnarray}
Similarly we define an $(\alpha,\beta)$-approximation to \textsc{MultiCKP}. 
\end{definition}
In particular, a {\em polynomial-time approximation scheme} (PTAS) is a $(1-\epsilon,1)$-approximation algorithm for any $\epsilon>0$.  The running time of a PTAS is polynomial in the input size for every fixed $\epsilon$, but the exponent of the polynomial may depend on $1/\epsilon$.  
An even stronger notion is a {\em fully polynomial-time approximation scheme} (FPTAS), which requires the running time to be polynomial in both input size and $1/\epsilon$. 
In this paper, we are interested in an FPTAS in the {\it resource augmentation model}, which is a $(1, 1+\epsilon)$-approximation algorithm for any $\epsilon>0$, with the running time being polynomial in the input size and $1/\epsilon$. We will refer to this as a $(1,1+\epsilon)$-FPTAS.

\subsection{Truthful Mechanisms}\label{sec:model.truth}

This section follows the terminology of \cite{NRTV07}.
We define truthful (aka. incentive-compatible) approximation mechanisms for our problem. We denote by $\cX\subseteq\cD^n$ the set of {\it feasible allocations} in our problem (\textsc{NsmCKP} or \textsc{Multi-$m$DKP}). 

\begin{definition}[Mechanisms]
Let $\cV\triangleq\cV_1\times\cdots\times\cV_n$, where $\cV_k$ is the set of all possible valuations of agent $k$. 
A mechanism $(\cA,\cP)$ is defined by an allocation rule $\cA:\cV\to\cX$ and a payment rule $\cP:\cV\to\RR^n_+$. We assume that the utility of player $k$, under the mechanism, when it receives the vector of bids $v\triangleq(v_1,\ldots,v_n)\in\cV$, is defined as $U_k(v)\triangleq\bar v_k(d_k(v))-p_k(v)$, where $\cA(v)=(d_1(v),\ldots,d_n(v)),$ and $\cP(v)=(p_1(v),\ldots,p_n(v))$ and $\bar v_k$ denotes the true valuation of player $k$.   
\end{definition}
Namely, a mechanism defines an allocation rule and payment scheme, and the utility of a player is defined as the difference between her valuation over her allocated demand and her payment.

\begin{definition}[Truthful Mechanisms]
A mechanism is said to be {\it truthful} if for all $k$ and all $v_k\in\cV_k$, and $v_{-k}\in\cV_{-k}$, it guarantees that $U_k(\bar v_k,v_{-k})\geq U_k(v_k,v_{-k})$. 
\end{definition}
Namely, the utility of any player is maximized, when she reports the true valuation. 
\begin{definition}[Social Efficiency]
A mechanism is said to be {\it $\alpha$-socially efficient} if for any $v\in\cV$, it returns an allocation $\bd\in\cX$ such that the  total valuation (also called {\it social welfare}) obtained is at least an $\alpha$-fraction of the optimum: $v(\bd)\ge\alpha \cdot 
\OPT$. 
\end{definition}
As in \cite{NR07,DN10,KTV13}, our truthful mechanisms are based on using {\it VCG payments} with {\it Maximal-in-Range} (MIR) allocation rules:
\begin{definition}[MIR]\label{d5}
An allocation rule $\cA:\cV\to\cX$ is an MIR, if there is a range $\cR\subseteq\cX$, such that for any $v\in\cV$, $\cA(v)\in\argmax_{\bd\in\cR}v(\bd)$.
\end{definition}
Namely, $\cA$ is an MIR if it maximizes the social welfare over a fixed ({\it declaration-independent}) range $\cR$ of feasible allocations. It is well-known (and also easy to prove by a VCG-based argument) that an MIR, combined with VCG payments (computed with respect to range $\cR$), yields a truthful mechanism. If, additionally, the range $\cR$ satisfies: $\max_{\bd\in\cR}v(\bd)\ge\alpha\cdot \max_{\bd\in\cX}v(\bd)$, then such a mechanism is also $\alpha$--socially efficient.    

Finally a mechanism is {\it computationally efficient} if it can be implemented in polynomial time (in the size of the input).

\subsection{Application to Power Allocation in Alternating Current Electric Systems} \label{sec:app}

Conventionally, the demands in (AC) alternating current electric systems are represented by active power in positive real numbers and reactive power in (positive or negative) imaginary real numbers, which are complex numbers in the first and fourth quadrants of the complex plane. We note that our problem is invariant, when the arguments of all demands are shifted by the same angle. For convenience of notation, we assume the demands are rotated by $90$ degrees unless all demands are entirely in the first quadrant of the complex plane\footnote{{  Note that it is customary in the power engineering literature to assume that the demands lie in the first and fourth quadrants of the complex plane. However, for convenience of presentation, we prefer to work in the first and second quadrants. For instance, if all the demands are capacitive (i.e., lie in the fourth quadrant), then rotation by $90$ degrees allows us to assume that the all numbers involved are {\it non-negative}, a property which is necessary for obtaining a PTAS.}}. See Fig.~\ref{fig:rotate} for a pictorial illustration. Notice that the axis labels are swapped after the rotation, the real axis indicates the reactive power while the imaginary indicates the active power.

\begin{figure}[!ht]
	\begin{center} 
		\includegraphics[scale=0.8]{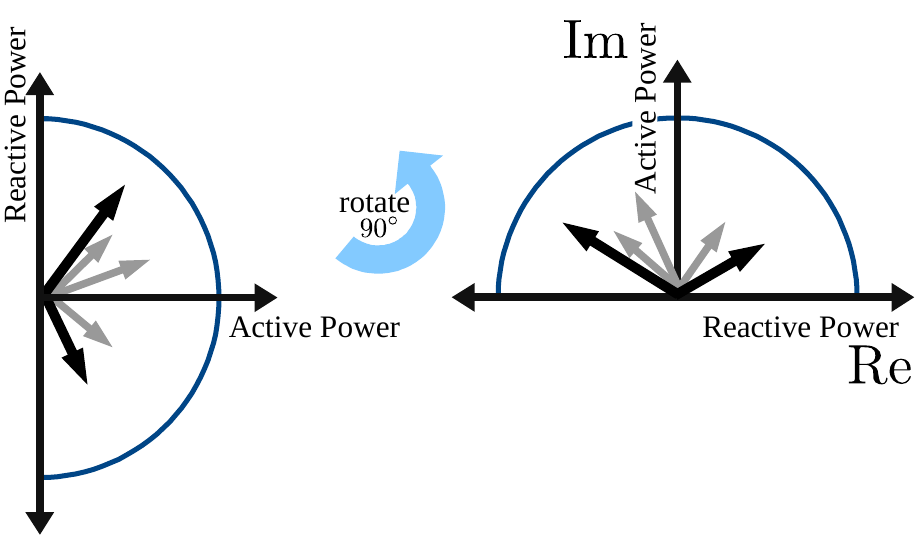}
	\end{center} 
	\caption{Each vector represents a power demand $d$. The figure shows that all demands are rotated by $90^\circ$. }
	\label{fig:rotate}
\end{figure} 

{\sc CKP} is a simplified model of the real-world AC electric systems, which considers the capacity constraint of a single link (e.g., a single bottleneck). There are practical scenarios, where the single-link capacity constraint is critical. For example, in a microgrid, there is usually a single transmission line connecting the main grid and the microgrid. In such a setting, our model can capture power allocation considering the capacity of transmission bottleneck.
We also remark that our single-link model is fundamental to general network setting with multiple capacitated links. A thorough understanding of the single-link case can pave the way to solving the multi-link case. As a follow-up study, our recent paper \cite{MCK16} considers non-strategic power allocation of power flow of inelastic demands in a power network, which is based on some of the fundamental results obtained in this paper.

In the conventional models of AC electric systems, there are multiple operating constraints, in addition to capacity constraint. One example is the nodal voltage constraint at customers. In our recent study of event-based demand response management in microgrids \cite{KKCEZ16}, we evaluate the changes of nodal voltage in response to power allocation decisions of customers. We observe that the voltage is less sensitive to the decisions of power allocation. Hence, {\sc CKP} is a suitable model for approximating the power allocation in such a setting.

We remark that the simplified DistFlow model \cite{baran1989sizing,baran1989placement,low2014convex1} is a well-known approximation model of power flows by ignoring the loss terms from the formulation. In fact, {\sc CKP} is equivalent to the simplified DistFlow model on a capacitated single link topology. 

Naturally, {\sc MultiCKP} models combinatorial power auctions in AC electric systems.
Each customer $k$ declares to the utility company a set of demands $D_k$ that represents her preferences among different alternatives of load profiles, and a  valuation function $v_k$ over $D_k$.  
The valuation $v_k(d)$ represents the amount customer $k$ is welling to pay if her load profile $d\in D_k$ is satisfied. 
If  customer $k$ wants to bid for a load that represents multiple appliances at once, she can include her corresponding vector sum to the set $D_k$ as an additional preference.
The monotonicity of $v_k$ with respect to the partial order ``$\preceq$'' implies the free disposal of extra supplied power (see Fig.~\ref{fig:monotone}). 
Larger active power (imaginary component) should have at least the value of smaller active power.
Conventionally, capacitive loads have negative reactive power, while inductive loads have positive reactive power.
Monotonicity implies that extra supplied reactive power of the same type (capacitive or inductive) should have at least the same valuation.
Indeed, customers can define constant valuation for different values of reactive power. 
We remark that all demands in the set $D_k$ are assumed to be either in the first quadrant or the second, but not in both. This in fact implies the load profiles of each customer are either only capacitive, or only inductive. Such mild assumption is actually needed in order to apply the $(1,1+\epsilon)$-FPTAS in Sec.~\ref{sec:tbp}.


\section{Hardness of Power Allocation in AC Electric Systems}\label{sec:hardness}


In this section, we present our main hardness result for \textsc{CKP}, which depends on the maximum angle $\phi$ the demands make with the positive real axis. When $\phi\in[\frac{\pi}{2}+\delta,\pi]$,  we show that the problem is inapproximable within any polynomial factor if we do not allow a violation of Constraint~\raf{C1}. 
Moreover, when $\phi$ approaches $\pi$, there is no $(\alpha,\beta)$-approximation, for any $\alpha$ and $\beta$ with polynomial bit length. 
Our hardness results indicate that the approximability of the problem \textsc{CKP} differs depending on maximum argument of any demand $\phi$. This insight suggests to study different techniques in the later sections to achieve the best approximation result possible for each case.

%

\

\begin{theorem}\label{hard}
Unless P=NP, for any $\delta>0$ and $\delta'>0$ 
\begin{itemize}
\item[(i)] there is no $(\alpha,1)$-approximation for \textsc{CKP$[\frac{\pi}{2}+\delta,\pi]$} where $\alpha,\delta$ have polynomial length.
\item[(ii)] there is no ($\alpha$, $\beta$)-approximation for \textsc{CKP$[\pi-\delta',\pi]$}, where $\alpha$ and $\beta$ have polynomial length, and $\delta'$ is exponentially small in $n$.
\end{itemize}

\end{theorem}

\

{\bf Remark.} In fact, these hardness results hold even if we assume that all demands are on the real line, except one demand $d_{m+1}$ such that ${\rm arg}(d_{m+1}) = \frac{\pi}{2}+\theta$, for some $\theta\in[\delta,\frac{\pi}{2}]$ (see Fig.~\ref{fig:theta}). Note that the trivial approximation of picking the user with the highest feasible value does not obtain $\tfrac{1}{n}$ approximation, because we allow demands to have both positive and negative real parts, which can cancel each other.

\begin{figure}[!htb]
	\begin{center}
		\includegraphics[scale=0.8]{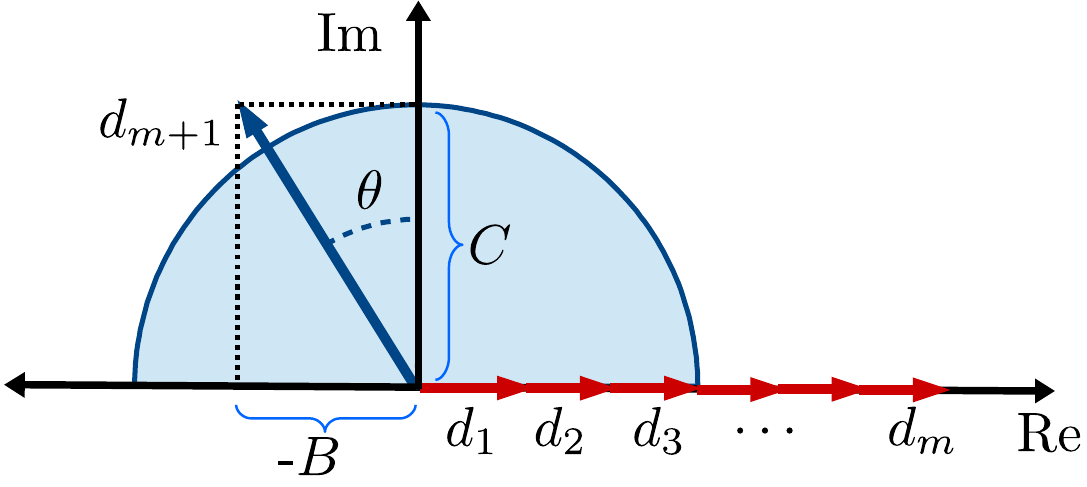}
	\end{center}
\caption{The set of demands $\{d_k\}$ for the proof of Theorem~\ref{hard}.}
	\label{fig:theta}
\end{figure}


\begin{proof}
We present a reduction from the (weakly) NP-hard subset sum problem ({\sc SubSum}): given an instance $I$, a set of positive integers ${A} \triangleq\{a_1,\ldots,a_m\}$ and a positive integer $B$, does there exist a subset of ${A}$ that sums-up to exactly $B$? 
We assume that $B$ is not polynomial in $m$, otherwise the problem can be easily solved in polynomial time by  dynamic programming.

We construct an instance $I'$ of  \textsc{CKP$[\frac{\pi}{2}+\theta,\frac{\pi}{2}+\theta]$} for each instance  $I$ of {\sc SubSum} such that if {\sc SubSum}$(I)$ is a ``yes'' instance then the optimum value of  \textsc{CKP$[\frac{\pi}{2}+\theta,\frac{\pi}{2}+\theta]$}, denoted by $\OPT$, is at least $1$; and if {\sc SubSum}$(I)$ is a  ``no" instance, then $\OPT < \alpha$ even when Cons.~\raf{C1} can be violated by $\beta$.

We define $n\triangleq m+1$ demands: for each $a_k, k= 1, ..., m$, define a demand $d_k \triangleq a_k$, and an additional demand
 $$d_{m+1} \triangleq -B + {\bf i} B \cot \theta.$$
For all $k = 1, ..., m$, let valuation $v_k \triangleq \frac{\alpha}{m+1}$, and $v_{m+1} \triangleq 1$. We let 
$$C \triangleq B \cot \theta.$$

We prove the first direction, assuming {\sc SubSum}($I$) is feasible. Namely, $\sum_{k = 1}^{m}a_k \hat{x}_k = B$, where $\hat x \in \{0,1\}^m$ is a solution vector of {\sc SubSum}. Construct a  solution $x\in \{0,1\}^{m+1}$ of \textsc{CKP} such that
\begin{equation*}
x_k = \left\{
\begin{array}{l l}
\hat x_k& \quad \text{if $k = 1, ..., m$}\\
1		& \quad \text{if $k = m+1$}.
\end{array} \right.
\end{equation*}
In fact, this is a feasible solution that satisfies Constraint~\raf{C1}: using $\sum_{k = 1}^{m}a_k \hat x_k - B = 0$,  we get
\begin{align*}
&\left(\sum_{k = 1}^{m} d_k^{\rm R} x_k + d_{m+1}^{\rm R} \right)^2 + \left(\sum_{k = 1}^{m} d_k^{\rm I} x_k+ d_{m+1}^{\rm I}  \right)^2\\
 &= \left(\sum_{k = 1}^{m} d_k^{\rm R} x_k - B  \right)^2 + B^2 \cot^2 \theta\\
&= B^2 \cot^2 \theta = C^2.
\end{align*}
Since $v_{m+1}=1$, the total value of such solution $v(x)\ge 1$, which implies that $\OPT$ is at least 1.

Conversely, assume that $\OPT\ge \alpha$. Let $x^*\in\{0,1\}^{m+1}$ be an optimal solution that may violate Cons.~\raf{C1} by   $\beta$.  Since user $m+1$ has valuation $v_{m+1}=1$, while the rest of users valuations total to less than $\alpha$: $\sum_{k = 1}^{m} v_k < \alpha$, user $m+1$  must be included in the optimum.
Therefore, substituting in Constraint~\raf{C1'},
\begin{equation*}
\left(\sum_{k = 1}^{m} d_k^{\rm R} x^*_k - B  \right)^2 + B^2 \cot^2 \theta \le \beta^2 C^2
\end{equation*} 
gives
\begin{align}
\left(\sum_{k = 1}^{m} a_k x^*_k - B  \right)^2 &\le \beta^2 C^2  - B^2 \cot^2 \theta \nonumber \\
&= B^2 \cot^2 \theta(\beta^2 - 1).  \label{h1}
\end{align}
By the integrality of the $a_i$'s, 
\begin{equation}\label{eq:h1.1}
\sum_{k= 1}^m a_k x^*_k = B \iff |\sum_{k= 1}^m a_k x^*_k - B|< 1
\end{equation}
In other words, {\sc SubSum} is feasible if and only if the absolute difference  $|\sum_{k= 1}^m a_k x^*_k - B|< 1$.  This implies,  {\sc SubSum}$(I)$ is feasible when the R.H.S. of Eqn. \raf{h1} is strictly less than  1. When $\beta = 1$, R.H.S. of Eqn. \raf{h1} is zero, and  we complete the second direction and hence, the proof of part (i) of the theorem.

For large enough $\theta$, the R.H.S. of Eqn.~\raf{h1} is strictly less than $1$:
$$ B^2 \cot^2 \theta(\beta^2 - 1) < 1$$
This implies, $\theta > \tan^{-1} \sqrt{B^2 (\beta^2 -1)}$. By 
Eqn.~\raf{eq:h1.1}, {\sc SubSum} is feasible which  completes the second direction and establishes part (ii) of the theorem.

\end{proof}

\section{A Truthful PTAS for {\sc MultiCKP}$[0,\frac{\pi}{2}-\delta]$}\label{sec:ckp1}
In this section, we present our truthful PTAS for {\sc MultiCKP}$[0,\frac{\pi}{2}-\delta]$. This PTAS invokes a truthful PTAS for \textsc{Multi-$m$DKP} as a subroutine.
Problem \textsc{Multi-$m$DKP} was shown in \cite{KTV13} to have a $(1-\epsilon)$-socially efficient truthful PTAS in the setting of {\it multi-unit auctions with a few distinct goods}, based on generalizing the result for the case $m=1$ in \cite{DN10}. We explain this result first in our setting, and then use it in Sections~\ref{sec:ptas} and \ref{sec:truthful-ptas} to derive a truthful PTAS for {\sc MultiCKP}$[0,\frac{\pi}{2}-\delta]$. We remark that, without the truthfulness requirement, our PTAS works even for $\delta=0$. However, we are only able to make it truthful for any given, but arbitrarily small, constant $\delta>0$.  Removing this technical assumption is an interesting open question.    

\subsection{A Truthful PTAS for \textsc{Multi-$m$DKP}}\label{sec:Tm-mC-KS}
 We present a truthful PTAS for \textsc{Multi-$m$DKP} that will be needed in Sec.~\ref{sec:ptas}-\ref{sec:truthful-ptas} below. This result is a slight generalization of \cite{KTV13} to accommodate real-valued demand vectors instead of integer-valued. 

Let $c=(c^1,\ldots,c^m)$ be the capacity vector, and for any $d\in\cD\subseteq\RR^m_+$, write $d_k=(d^1_k,\ldots,d^m_k)$. For any subset of users $N\subseteq\cN$ and a partial selection of demands $\bar\bd=(d_k\in\cD:~k\in N)$, such that $\sum_{k\in N}d_k\le c$, define the vector $b_{N,\bar\bd}=(b_{N,\bar\bd}^1,\ldots,b_{N,\bar\bd}^m)\in\RR^m_+$ as follows
\begin{equation}\label{bdT}
b_{N,\bar\bd}^i=\frac{c^i-\sum_{k\in N}d_k^i}{(n-{|N|})^2}.
\end{equation}
Following \cite{DN10,NR07,KTV13}, we consider a restricted range of allocations defined as follows: 
\begin{equation}\label{range}
\cS\triangleq\bigcup_{\stackrel{N\subseteq\cN,~\bar\bd=(d_k:~k\in N):~|N|\le\frac{m}{\epsilon},}{d_k\in\cD~\forall k\in N}}\cS_{N,\bar\bd},
\end{equation}
where, for a set $N\subseteq\cN$ and a partial selection of demands $\bar\bd=(\bar d_k\in \cD:~k\in N)$,
{\small
\begin{align*}
\cS_{N,\bar\bd}&\triangleq\Big\{(d_1,\ldots, d_n)\in\cD^n~|~\sum_{k\in\cN}d_k\le c, d_k=\bar d_k\ \forall k\in N,\\
&\forall k\not\in N~\forall i\in[m]~\exists r_k^i\in\ZZ_+ \mbox{\ s.t.\ }  d_k^i=r_k^i\cdot b_{N,\bar\bd}^i\text{ and } \sum_{k\not\in N}r_k^i\le (n-|N|)^2~\Big\}.
\end{align*}}
\hspace{-0.05in}Note that the range $\cS$ {\it does not} depend on the declarations $D_1,\ldots,D_n$. The following two lemmas establish that the range $\cS$ is a good approximation of the set of all feasible allocations and that it can be optimized over in polynomial time. The first lemma is essentially a generalization of a similar one for multi-unit auctions in \cite{DN10}, with the simplifying difference that we do not insist here on demands to be integral. The second lemma is also a generalization of a similar result in \cite{DN10}, which was stated for the multi-unit auctions with a few distinct goods in \cite{KTV13}. For completeness, we give the proofs in the appendix.  
\begin{lemma}[\cite{DN10}]\label{l1-}
$\max_{\bd\in\cS}v(\bd)\ge(1-\epsilon)\OPT.$
\end{lemma}

\begin{lemma}[\cite{DN10,KTV13}]\label{l2-}
We can find $\bd^*\in\argmax_{\bd\in\cS}v(\bd)$ using dynamic programming in time $\left|\bigcup_{k}D_k\right|^{O(m/\epsilon)}$.
\end{lemma}

It follows that an allocation rule defined as an MIR over range $\cS$ yields a $(1-\epsilon)$-socially efficient truthful mechanism for \textsc{Multi-$m$DKP}.


\subsection{A PTAS for {\sc MultiCKP}$[0,\frac{\pi}{2}]$} \label{sec:ptas}
We now apply the result in the previous section to the multi-minded complex-demand knapsack problem, when all agents are restricted to report their demands in the positive quadrant. We begin first by presenting a PTAS without strategic considerations; then it is shown in the next section how to use this PTAS within the aforementioned framework of MIR's to obtain a truthful mechanism.

\ 

{\bf Overview of Technique.}  
As we will see in Sec.~\ref{sec:tbp}, it is possible to obtain a $(1,1+\epsilon)$-approximation by a reduction to the {\sc Multi-$2$DKP} problem. To get a better result without violating the constraint, we reduce {\sc MultiCKP}$[0,\frac{\pi}{2}]$ instance to {\sc Multi-$m$DKP}.
We note that although there is a PTAS for {\sc Multi-$m$DKP} with constant $m$ \cite{FC84alg}, such a PTAS cannot be directly applied to \textsc{MultiCKP$[0,\frac{\pi}{2}]$} by polygonizing the circular feasible region for \textsc{MultiCKP$[0,\frac{\pi}{2}]$}, because one can show that such an approximation ratio is at least a constant factor. This is the case, for instance, if the optimal solution consists of a few large (in magnitude) demands together with many small demands, and it is not clear at what level of accuracy we should polygonize the region to be able to capture these small demands. To overcome this difficulty, we have to first guess the large demands, then we construct a grid (or a lattice) on the remaining part of the circular region, defining a polygonal region in which we try to pack the maximum-utility set of demands. The latter problem is easily seen to be a special case of the {\sc Multi-$m$DKP} problem. The main challenge is to choose the granularity of the grid small enough to well-approximate the optimal, but also large enough so that the number of sides of the polygon, and hence $m$ is a constant only depending on $1/\epsilon$.  

\

In this section  we assume that ${\rm arg}(d) \le \frac{\pi}{2}$, that is, $d^{\rm R} \ge 0$ and $d^{\rm I} \ge 0$ for all $d\in\cD$. 
Without loss of generality, we assume $\epsilon<\frac{1}{4}$ where $\frac{1}{\epsilon}\in\ZZ_+$.
For an integer $i\in\ZZ_+$, let $\cL_1(i)$ and $\cL_2(i)$, respectively, denote the sets of all vertical and all horizontal lines in the complex plane that are at (non-negative) distances, from the real and imaginary axes, which are integer multiples of $\frac{C}{2^i}$, that is,
\begin{eqnarray*}
\cL_1(j)&\triangleq&\{x+{\bf i}y\in \CC~|~x=\frac{\lambda C}{2^j},~\lambda\in \ZZ_+,y\in\RR\},\\
\cL_2(j)&\triangleq&\{x+{\bf i}y\in \CC~|~y=\frac{\lambda C}{2^j},~\lambda\in \ZZ_+,x\in\RR\}.
\end{eqnarray*}

Given a feasible set of vectors $T\subseteq \cD$ to \textsc{MultiCKP$[0,\frac{\pi}{2}]$} (that is, $\left|\sum_{d\in T}d\right|\le C$), define $d_T\triangleq\sum_{d \in T} d$, and let
\begin{equation}\label{wT}
w^{\rm I}_T \triangleq \sqrt{C^2 -{\rm Re}(d_T)^2} - {\rm Im}(d_T),~ w^{\rm R}_T \triangleq \sqrt{C^2 -{\rm Im}(d_T)^2} - {\rm Re}(d_T).
\end{equation}
Namely,  $w^{\rm R}_T$ (resp.  $w^{\rm I}_T$) is the horizontal (resp. vertical) distance between $d_T$ and the boundary of the constraint disk (see Fig.~\ref{f1} for an illustration).

Let $\rho_1(T)$ and $\rho_2(T)$ be the smallest integers such that 
$$
{C}\big/{2^{\rho_1(T)}}\le\frac{\epsilon w_T^{\rm R} }{4} \text{ and } {C}\big/{2^{\rho_2(T)}}\le\frac{\epsilon w_T^{\rm I}}{4}. 
$$

The set of lines in $\cL_1(\rho_1(T))\cup\cL_2(\rho_2(T))$ define a grid on the feasible region at ``vertical and horizontal levels'' $\rho_1(T)$ and  $\rho_2(T)$, respectively. We observe that if we increase $d^{\rm R}_T$, then $w^{\rm R}_T$ decreases as well as the real granularity of the grid, and hence  $\rho_1(T)$ becomes larger. Similar observation holds when we increase $d^{\rm I}_T$
 
Let $\lambda_1(T)$ and $\lambda_2(T)$ be the largest integers such that
$$
d_T^{\rm R}\ge \frac{\lambda_1(T) C}{2^{\rho_1(T)}}~\text{ and }~d_T^{\rm I}\ge \frac{\lambda_2(T) C}{2^{\rho_2(T)}},
$$
and $z_T\in\CC$ be the intersection of the two lines corresponding to $\lambda_1(T)$ and $\lambda_2(T)$: 
$$
z_T\triangleq\{x+{\bf i}y\in \CC~|~x=\frac{\lambda_1(T) C}{2^{\rho_1(T)}}\}\cap\{x+{\bf i}y\in \CC~|~y=\frac{\lambda_2(T) C}{2^{\rho_2(T)}}\}.
$$
Given $z_T$, we define four points in the complex plane $({\psi'}_T^1,\psi_T^1,\psi_T^2,{\psi'}_T^2)$ such that
{\small
\begin{eqnarray*}
& {\psi'}_T^1 = \Big(0, \sqrt{C^2 -{\rm Re}(z_T)^2} \Big), ~ \psi_T^1 = \Big({\rm Re}(z_T),\sqrt{C^2 -{\rm Re}(z_T)^2}\Big), \\
& {\psi'}_T^2 = \Big(\sqrt{C^2 -{\rm Im}(z_T)^2}, 0\Big), ~ \psi_T^2 = \Big(\sqrt{C^2 -{\rm Im}(z_T)^2},{\rm Im}(z_T)\Big).
\end{eqnarray*}}
Let $\cR_T$ be the part of the feasible region dominating $z_T$:
\begin{equation*}
\cR_T\triangleq\{x+{\bf i}y\in\CC~:~|x+{\bf i}y|\leq C,~x\ge {\rm Re}(z_T), y\ge {\rm Im}(z_T)\},
\end{equation*}
and $P_T(\epsilon)$ be the set of intersection points\footnote{For simplicity of presentation, we will ignore the issue of finite precision needed to represent intermediate calculations (such as the square roots above, or the intersection points of the lines of the grid with the boundary of the circle).}
between the grid lines in $\cL_1(\rho_1(T))\cup\cL_2(\rho_2(T))$ and the boundary of $\cR_T$:
$$
P_T(\epsilon)\triangleq\{z\in\cR_T~:~|z|= C\}\cap(\cL_1(\rho_1(T))\cup\cL_2(\rho_2(T))).
$$
The convex hull of the set of points $P_T(\epsilon)\cup\{{\psi'}_T^1,\psi_T^1,\psi_T^2,{\psi'}_T^2, \bzero\}$ defines a polygonized region, which we denote by $\cP_T(\epsilon)$ and its size (number of sides) by $m_T(\epsilon)$ (see Fig.~\ref{f1} for an illustration).   

\begin{lemma}\label{l-size}
$ m_T(\epsilon) \leq \frac{18}{\epsilon}+3$. 
\end{lemma}
\begin{proof}
Let $x$ be the horizontal distance from $z_T$ to the boundary of the circle (with center $\bzero$ and radius $C$). Then, by the definition of $\rho_1(T)$,
$
\frac{C}{2^{\rho_1(T)-1}}>\frac{\epsilon w_T^{\rm R} }{4},
$
which implies that $\frac{w_T^{\rm R} }{C}<\frac{1}{2^{\rho_1(T)-3}\epsilon}$. On the other hand, by the definition of $z_T$,
$$
x\le w_T^{\rm R}+\frac{C}{2^{\rho_1(T)}}\le w_T^{\rm R}+\frac{\epsilon w_T^{\rm R} }{4}=\left(1+\frac{\epsilon}{4}\right)w_T^{\rm R} . 
$$
It follows from the above inequalities that the number of vertical grid lines (at level $\rho_1(T)$) between $z_T$ and the boundary of the circle is at most
$$
\frac{x}{C/2^{\rho_1(T)}}+1\le\frac{2^{\rho_1(T)}\left(1+\frac{\epsilon}{4}\right)w_T^{\rm R}}{C}+1
<\frac{8\left(1+\frac{\epsilon}{4}\right)}{\epsilon}+1<\frac{9}{\epsilon}.
$$
Similarly, we can show that the number of horizontal grid lines between $z_T$ and the boundary of the circle is at most $\frac{9}{\epsilon}$. Adding the three other points $\{{\psi'}_T^1,{\psi'}_T^2, \bzero\}$ gives the claim. 
\end{proof}

	\begin{figure}[!htb]
		\centering
		\includegraphics[scale=0.8]{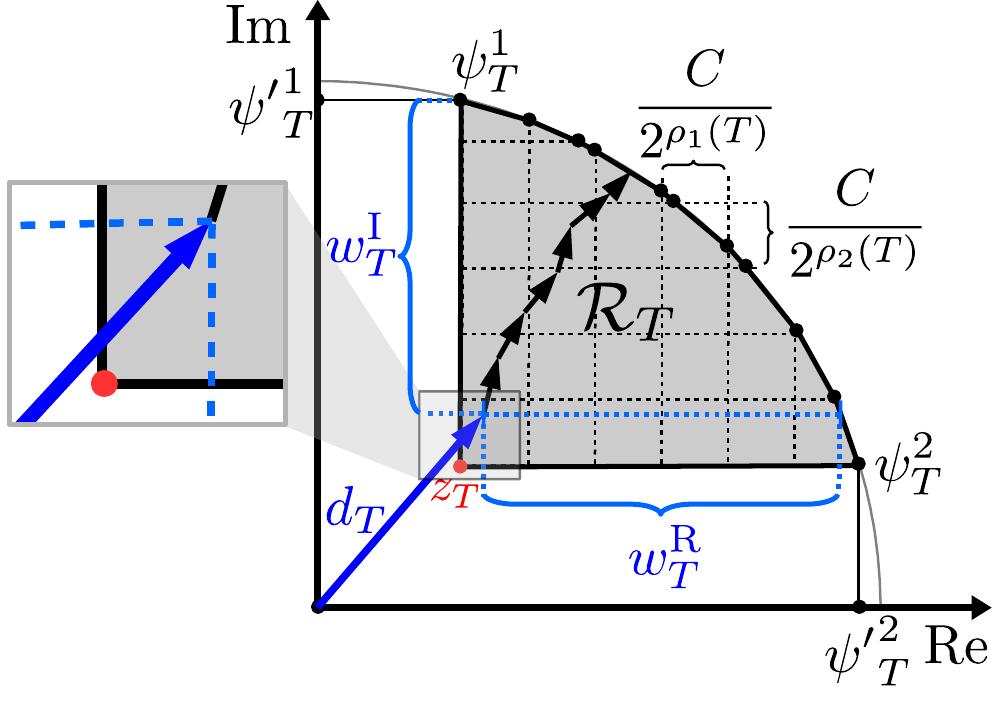}
		\caption{We illustrate the region $\cR_T$ by the shaded area and $P_T(\epsilon)$ by the black dots on the arc of the circle. The zoomed area highlights that $w^{\rm R}_T$ (resp. $w^{\rm I}_T$) is defined by $d_T$ {\em not} $z_T$.}
		\label{f1}
		\end{figure}

	%

\begin{definition}
Consider a subset of users $N\subseteq \cN$ and a feasible set $T\triangleq\{\overline d_k:k\in N\}$ to {\sc MultiCKP}$[0,\frac{\pi}{2}]$. We define an approximate problem ({\sc PGZ}$_T$) by polygonizing \textsc{MultiCKP$[0,\frac{\pi}{2}]$}:
\begin{eqnarray*}
\textsc{(PGZ$_T$)} \qquad& \displaystyle \max  \sum_{k\in\cN}v_k (d_k) \label{PGZ}\\
\text{s.t.}\qquad & \displaystyle \sum_{k\in \cN}d_k \in \cP_T(\epsilon) \label{CPGZ}\\
\qquad & d_k=\overline d_k,~~~~\forall k\in N\label{DPGZ}\\ 
\qquad & d_k\in\cD, ~~~~\forall k\in\cN\backslash N.
\end{eqnarray*}
\end{definition}

Given two complex numbers $\mu$ and $\nu$, we denote the projection of $\mu$ on $\nu$ by ${\tt Pj}_\nu(\mu) \triangleq \frac{\nu}{|\nu|^2}(\mu^{\rm R}\nu^{\rm R} +  \mu^{\rm I}\nu^{\rm I})$. Given the convex hull $\cP_T(\epsilon)$, we define a set of $m_T(\epsilon)$ vectors $\{ \sigma_T^i\}$, each of which is perpendicular to each boundary edge of $\cP_T(\epsilon)$ and starting at the origin (see Fig.~\ref{f2} for an illustration).  

	\begin{figure}[!htb]
		\centering
		\includegraphics[scale=0.8]{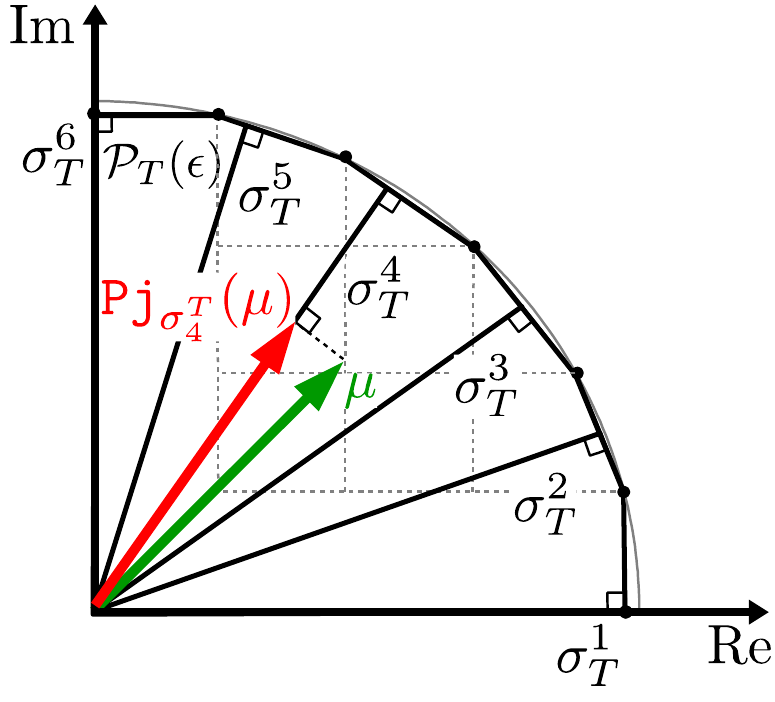}
		
		\caption{Each $\sigma_T^i$ is a vector (starting at the origin) perpendicular to each boundary edge of $\cP_T(\epsilon)$.\\ }
		
		\label{f2}
	\end{figure}

\begin{definition}
Consider a subset of users $N\subseteq \cN$ and a feasible set $T\triangleq\{\overline d_k:k\in N\}$ to {\sc MultiCKP}$[0,\frac{\pi}{2}]$. We define a \textsc{Multi-$m$DKP} problem based on $\{ \sigma_T^i\}$:

\begin{eqnarray}
\textsc{(Multi-$m$DKP$\{ \sigma_T^i\}$)}&\quad \displaystyle \max  \sum_{k\in\cN}v_k (d_{k}) \label{(mDKP)}\\
\text{s.t.}& \displaystyle \sum_{k\in\cN} | \Pj(d_k) | \le  |\sigma_T^i|, \quad \forall i= 1,\ldots, m_T(\epsilon) \label{mCm-1}\\
\qquad & d_k=\overline d_k,~~~~\forall k\in N\label{mCm-3}\\ 
\qquad & d_k\in\cD, ~~~~\forall k\in\cN\backslash N. \label{mCm-2}
\end{eqnarray}
One can see that \textsc{Multi-$m$DKP$\{ \sigma_T^i\}$} is an instance of {\sc Multi-$m$DKP} (defined in Sec.~\ref{MCMDKS-sec}) over the set of users $\cN\backslash N$,  $m=m_T(\epsilon)$, $d^i_k = | \Pj(d_k) |$, and $c^i = |\sigma_T^i| - \sum_{k \in N} | \Pj(\overline d_k) |$.

\end{definition}

\begin{lemma}\label{lem-proj}
Given a feasible set $T$ to \textsc{MultiCKP$[0,\frac{\pi}{2}]$},  {\sc PGZ}$_T$ and {\sc Multi-$m$DKP}$\{ \sigma_T^i\}$ are equivalent.
\end{lemma}
Lemma~\ref{lem-proj} follows straightforwardly from the convexity of the polygon $\cP_T(\epsilon)$.

Our PTAS for \textsc{MutliCKP$[0,\frac{\pi}{2}]$} is described in Algorithm {\sc MultiCKP-PTAS}, which enumerates every subset partial selection $T$ of at most $\frac{1}{\epsilon}$ demands, then finds a near optimal allocation for each polygonized region $\cP_T(\epsilon)$ using the PTAS of {\sc Multi-$m$DKP} from Section~\ref{sec:Tm-mC-KS}, which we denote by {\sc Multi-$m$DKP-PTAS}$[\cdot]$.
\begin{algorithm}[!htb]
	\caption{ {\sc MultiCKP-PTAS}$(\{v_k,D_k\}_{k\in\cN},C,\epsilon)$} \label{CKP-PTAS}
\begin{algorithmic}[1]
\Require Users' multi-minded valuations $\{v_k,D_k\}_{k\in \cN}$; capacity $C$; accuracy parameter $\epsilon$
\Ensure $(1-3\epsilon)$-allocation $(\widehat{d}_1,\ldots,\widehat d_n)$ to \textsc{MultiCKP$[0,\frac{\pi}{2}]$}
\State $(\widehat{d}_1,\ldots,\widehat d_n) \leftarrow (\bzero,\ldots,\bzero)$
\For{each subset $N\subseteq \cN$ and each subset $T=\{\overline d_k\in D_k:k\in N\}$ of size at most $\frac{1}{\epsilon}$ s.t. $\big|\sum_{d\in T}d\big|\le C$}\label{ss1}
  \State Set $d_T \leftarrow \sum_{d\in T}d$, and define the corresponding vectors $\{ \sigma_T^i\}$
  \State Obtain $(d_1,\ldots,d_n) \leftarrow$ {\sc Multi-$m$DKP-PTAS} [{\sc Multi-$m$DKP}$\{ \sigma_T^i\}$] within accuracy $\epsilon$ \label{s1} 
  \If{$\sum_kv_k(\widehat{d}_k) < \sum_kv_k(d_k)$}
  \State $(\widehat{d}_1,\ldots,\widehat d_n) \leftarrow (d_1,\ldots,d_n)$
  \EndIf
\EndFor
\State \Return $(\widehat{d}_1,\ldots,\widehat d_n)$
\end{algorithmic}
\end{algorithm}

\begin{theorem}\label{t2}
For any $\epsilon>0$, Algorithm {\sc MultiCKP-PTAS} finds a $(1-3\epsilon, 1)$-approximation to \textsc{MultiCKP$[0,\frac{\pi}{2}]$}. 
The running time of the algorithm is $\left|\bigcup_k D_k\right|^{O(\frac{1}{\epsilon^2})}$.
\end{theorem}
\begin{proof}
First, the upper bound on the running time of Algorithm {\sc MultiCKP-PTAS} is due to the fact that each of the $\left|\bigcup_k D_k\right|^{O\left (\frac{1}{\epsilon}\right)}$ iterations in line~\ref{ss1} requires invoking the PTAS of {\sc Multi-$m$DKP}, which in turn takes $\left|\bigcup_{k}D_k\right|^{O(m/\epsilon)}$ time, by Lemma~\ref{l2-}, where $m =O(\frac{1}{\epsilon})$. Therefore  the total running time is $\left|\bigcup_{k}D_k\right|^{O(1/\epsilon)} \cdot \left|\bigcup_{k}D_k\right|^{O(1/\epsilon^2)}  =  \left|\bigcup_{k}D_k\right|^{O(1/\epsilon^2)}$.

The algorithm outputs a feasible allocation by Lemma \ref{lem-proj} and the construction of $\cP_T(\epsilon)$.
To prove the approximation ratio, we show in Lemma~\ref{main-lem} below that, for any optimal (or feasible) allocation $(d_1^*,\ldots,d_n^*)$, we can construct another feasible allocation $(\widetilde d_1,\ldots,\widetilde d_n)$ such that $\sum_{k}v_k(\widetilde d_k)\ge(1-2\epsilon)\sum_kv_k(d_k^*)$ and $(\widetilde d_1,\ldots,\widetilde d_n)$ is feasible to {\sc PGZ$_T$} for some $T$ of size at most $\frac{1}{\epsilon}$. By Lemma \ref{lem-proj}, invoking the PTAS of {\sc Multi-$m$DKP$\{\sigma^i_T\}$} gives  a $(1-\epsilon)$-approximation $(\widehat d_1,\ldots,\widehat d_k)$ to {\sc PGZ$_T$}. Then
{\small
\begin{eqnarray*}
\sum_kv_k( \widehat d_k )\ge (1-\epsilon) \sum_kv_k(\widetilde d_k) \ge  (1-3\epsilon) \OPT.
\end{eqnarray*}}
\hspace{-0.05in}We give an explicit construction of the allocation $(\widetilde d_1,\ldots,\widetilde d_n)$ in Algorithm~\ref{Construct}, thus completing the proof by Lemma~\ref{main-lem}.
\end{proof}

\begin{lemma}\label{main-lem}
Consider a feasible allocation $\bd=(d_1,\ldots,$ $d_n)$ to \textsc{MultiCKP$[0,\frac{\pi}{2}]$}. Then we can find a set $T\subseteq \{d_1,\ldots,$ $d_n\}$ and construct an allocation $\widetilde \bd=(\widetilde d_1,\ldots,\widetilde d_n)$, such that $|T|\le\frac{1}{\epsilon}$ and
$\widetilde\bd$ is a feasible solution to {\sc PGZ}$_T$ and $v(\widetilde\bd) \ge(1-2\epsilon)v(\bd)$.    
\end{lemma}
\begin{proof}
In Algorithm~\ref{Construct}, let $\bar \ell$ and $T_{\bar \ell}$ be the values of $\ell$ and $T_\ell$ at the end of the repeat-until loop (line~\ref{ss1-}).

The basic idea of Algorithm~\ref{Construct} is that we first construct a nested sequence of sets of demands $T_0 \subset T_1 \subset\ldots\subset T_{\bar \ell}$, such that a demand is included in each iteration if it has either a large real component or a large imaginary component. The iteration proceeds until a sufficiently large number of demands have been summed up (namely, $|T_{\bar \ell}|\ge \frac{1}{\epsilon}$), or no demands with large components remain. At the end of the iteration, if the condition in line~\ref{sss1} holds, then $S = {T_\ell}$, i.e., the whole set $S$ can be packed within the polygonized region $\cP_{T_{\bar\ell}}(\epsilon)$. Otherwise, we find a subset of $S$ that is feasible to {\sc PGZ}$_{T_{\bar\ell}}$.

To do so, we partition $S \backslash T_{\bar\ell}$ into groups (possibly one group), each having a large component along either the real or the imaginary axes, with respect to the boundaries of the region $\cR_{T_{\bar \ell}}$. Then removing the group with smallest value among these, or removing one of the large demands with smallest valuation  will ensure that remaining demands have a large value and can be packed within $\cP_{T_{\bar\ell}}(\epsilon)$. 

We then have to consider two cases (line~\ref{sss2}): (i) $|T_{\bar \ell}|$ becomes at least $\frac{1}{\epsilon}$, or (ii) ${S}_{\bar \ell}^{\rm R}\cup {S}_{\bar \ell}^{\rm I} = \varnothing$. 
For case (i), We reduce the size of $T_{\bar \ell}$ by taking the first $\frac{1}{\epsilon}$ demands, and then we combine the remaining demands in $S\backslash T_{\bar \ell}$ into a group $V_1$. We show next that removing any one demand $d_k \in T_{\bar \ell}$ will make $S \backslash \{d_k\}$ a feasible solution to {\sc PGZ}$_{T_{\bar\ell}\backslash\{d_k\}}$. Since $d_{T_\ell} \succeq d_{T_{\ell'}}$ for $\ell> \ell'$,  the lengths $w_{T_\ell}^{\rm R}$ and $w_{T_\ell}^{\rm I}$ are monotone decreasing for $\ell=1,2,\ldots$ Indeed, w.l.o.g., let $\ell$ be such that $d_{k}\in S^{\rm R}_{\ell}$ (also $d_k \not\in T_\ell$) and let $T=T_{\bar\ell}\backslash\{d_k\}$.  Then, $w_T^{\rm R}\le w_{T_\ell}^{\rm R}$ implies that 
$$
\frac{C}{2^{\rho_1(T)}}\le\frac{C}{2^{\rho_1(T_\ell)}}\le\frac{\epsilon w_{T_\ell}^{\rm R}}{4}<d_k^{\rm R}.
$$
Hence, $d^{\rm R}_k$ is larger than the width of the grid's cell $\frac{C}{2^{\rho_1(T)}}$. Therefore, $d_{S\backslash \{d_k\}}$ is a feasible solution to $\text{{\sc PGZ}}_{T}$ because the removed demand exceeds the width of the grid (see Fig.~\ref{fig:rm} for an illustration).

\begin{figure}
	\centering
	\includegraphics[scale=0.8]{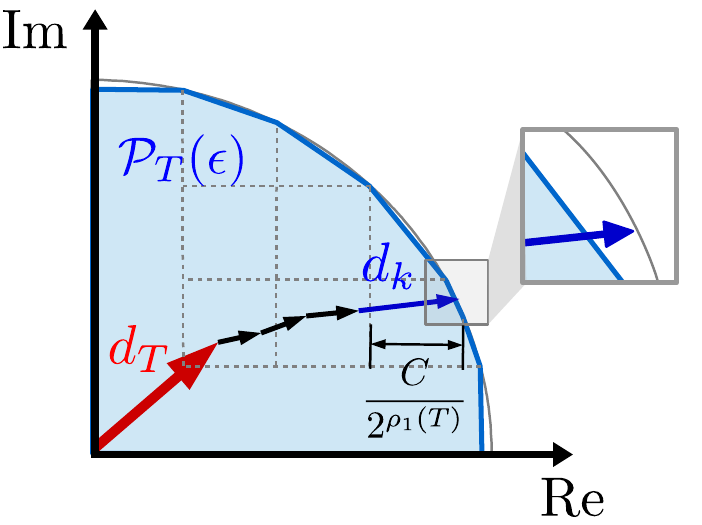}
	\caption{$d^{\rm R}_k$ is larger then the grid's width. If $d_k$ is removed, the remaining set of demands is feasible to $\text{{\sc PGZ}}_{T}$. }
	\label{fig:rm}
\end{figure}

For case (ii), we can apply Lemma~\ref{lem-pack} below to partition $S\backslash T_{\bar \ell}$ into at least $\frac{1}{\epsilon}-1$ groups  $\{V_1,\ldots, V_h\}$, where each group $V_j$ has a large total component along either the real or the imaginary axes (precisely, greater than $\frac{\epsilon}{4}w^{\rm R}_{T_{\bar\ell}}$ or $\frac{\epsilon}{4}w^{\rm I}_{T_{\bar\ell}}$ respectively). This implies that removing any group $V_j$ will make $T_{\bar\ell} \cup \bigcup_{j' \ne j} V_{j'}$ a feasible solution to {\sc PGZ}$_{T_{\bar\ell}}$.

To conclude, there are either (i) at least $\frac{1}{\epsilon}$ demands in $|T_{\bar \ell}|$, or (ii) ${S}_{\bar \ell}^{\rm R}\cup {S}_{\bar \ell}^{\rm I} = \varnothing$. 
We define $S'$ by deleting a minimum valuation demand  or group of demands from $ S$ (lines \ref{sss4} and \ref{sss5}). Then, we set $\widetilde d_k=d_k$ if $k\in S'$ and  $\widetilde d_k=0$ if $k\not \in S'$. Hence, in case (i), $v(\widetilde\bd)\ge(1-\epsilon)\OPT$, and in case (ii), $v(\widetilde\bd)\ge (1-\frac{1}{h})\OPT\ge\frac{1-2\epsilon}{1-\epsilon}\cdot\OPT\geq(1-2\epsilon)\OPT$. 
\end{proof}

\begin{algorithm}[!htb]
\caption{{\sc Construct}$(\{v_k\}_{k\in\cN},\bd,C,\epsilon)$}
\label{Construct}
\begin{algorithmic}[1]
\Require  \mbox{Users' valuations $\{v_k\}_{k\in\cN}$}; \mbox{a feasible allocation $\bd =(d_1,...,d_n)$};  \mbox{capacity $C$};
\Statex  \hspace{30pt} \mbox{accuracy parameter $\epsilon$}
\Ensure A set of demands $T\subseteq\{d_1,\ldots,d_n\}$ and a feasible allocation $\widetilde{\bd}$ 
\Statex \hspace{-18pt}{\bf Initialize:} $S\leftarrow\{d_1,\ldots,d_n\}$; $\widetilde\bd=\bd$; $\ell\leftarrow 0$; $T_\ell\leftarrow\varnothing$;  $\eta_\ell \leftarrow \bzero$ \label{init}
\Statex \Comment{{\em Find a subset of large demands $T_{\ell}$}}
\Repeat 
  \State $\ell\leftarrow\ell+1$
  \State $d_{T_\ell} \leftarrow \sum_{d\in {T_{\ell-1}}}d$
  \State $w^{\rm I}_{T_\ell} \leftarrow \sqrt{C^2 -{\rm Re}(d_{T_\ell})^2} - {\rm Im}(d_{T_\ell}); \quad 
  w^{\rm R}_{T_\ell} \leftarrow \sqrt{C^2 -{\rm Im}(d_{T_\ell})^2} - {\rm Re}(d_{T_\ell})$
  \State ${S}_\ell^{\rm R}\leftarrow\{d\in S\setminus T_\ell\mid d^{\rm R}>\frac{\epsilon}{4} w^{\rm R}_{T_\ell}\}$; \quad ${S}_\ell^{\rm I}\leftarrow\{d\in S\setminus T_\ell\mid d^{\rm I}>\frac{\epsilon}{4} w^{\rm I}_{T_\ell}\}$
  \State ${T_\ell} \leftarrow {T_\ell} \cup{S}_\ell^{\rm R}\cup {S}_\ell^{\rm I}$ 
  \State $\eta_{\ell}\leftarrow \eta_{\ell-1}+\sum_{d\in {S}_\ell^{\rm R}\cup {S}_\ell^{\rm I}}d$
  \Until{ $|{T_\ell}|\ge \frac{1}{\epsilon}$ \  or \   ${S}_\ell^{\rm R}\cup {S}_\ell^{\rm I}=\varnothing$ \   or  \  $ S\backslash {T_\ell}=\varnothing$}\label{ss1-}   
 
\State $\kappa\leftarrow\sum_{d\in  S\backslash {T_\ell}}d$ 
\If{$S\backslash {T_\ell}=\varnothing$ \   or \   $\eta_{\ell}+\kappa\in \cP_{T_\ell}(\epsilon)$} \label{sss1} 
  \State \Return $({T_\ell},\bd)$
\Else
       \Statex \Comment{{\em Find a subset $S' \subset S$  that is feasible to {\sc PGZ}$_{T_{\ell}}$}}
     \If{$|{T_\ell}|\ge\frac{1}{\epsilon}$} \label{sss2}
         \State $T_\ell\leftarrow $ the set of the first $\frac{1}{\epsilon}$ elements added to $T_\ell$ 
         \State $h\leftarrow1$; $V_1\leftarrow S\backslash {T_\ell}$\label{sss2-}
      \Else

        \State Find a partition $V_1,\ldots, V_h$ over $S \backslash {T_\ell}$ such that either 
		\Statex \qquad \qquad (i)~ $\sum_{d\in V_j}d^{\rm R}\geq\frac{\epsilon}{4} w^{\rm R}_{T_\ell}$ for all $j={1,\ldots,h}$, or 
		\Statex \qquad \qquad (ii) $\sum_{d\in V_j}d^{\rm I}\ge\frac{\epsilon}{4}w^{\rm I}_{T_\ell} $ for all $j={1,\ldots,h}$,   \label{sss3}
		\Statex \qquad \quad where $h$ is defined in Lemma~\ref{lem-pack}
     \EndIf
	 \State Pick $\widehat{k}\in\argmin\{v_k(d_k)\mid d_k\in {T_\ell}\}$\label{sss3-}
	 \State Pick $\widehat{j}\in\argmin\{\sum_{k:d_k\in V_j}v_k(d_k)\mid j={1,\ldots,h}\}$
      \If{$v_{\widehat{k}}(d_{\widehat{k}}) < \sum_{k:d_k\in V_{\widehat{j}}}v_k(d_k)$}\label{sss3.4}
          \State  $\widetilde d_{\widehat k}\leftarrow\bzero$
          \State \Return $(T_\ell\backslash\{d_{\widehat k}\},\widetilde\bd)$\label{sss4}
      \Else
          \State $\widetilde d_k\leftarrow \bzero$ for all $k:d_k\in V_{\widehat{j}}$ 
          \State \Return $(T_\ell, \widetilde\bd)$\label{sss5}
      \EndIf
\EndIf
\end{algorithmic}
\end{algorithm}

\begin{figure}
\centering
  \includegraphics[scale=0.8]{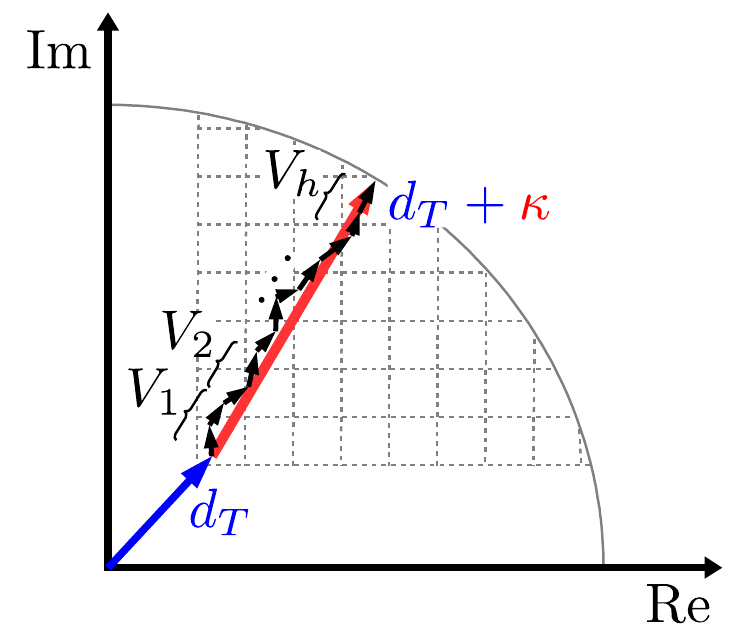}
	\caption{Demands $\{d_k \mid T \in \cN \backslash T\}$ are partitioned into $\{V_i\}$. }
	\label{f3}
\end{figure}

\begin{lemma}\label{lem-pack}
Consider a set  of demands $S\subseteq\cD$ and $T \subseteq S$, such that
\begin{enumerate}
\item $S$ is feasible solution to \textsc{MultiCKP$[0,\frac{\pi}{2}]$}, but $S$ is not a feasible solution to {\sc PGZ}$_T$
\item $d^{\rm R}\le \frac{\epsilon}{4} w^{\rm R}_T$ and  $d^{\rm I}\le \frac{\epsilon}{4} w^{\rm I}_T$, for all $d \in S \backslash T$. 
\end{enumerate}
 Then there exists a partition $\{V_1,\ldots, V_h\}$ of $S \backslash T$ such that 
\begin{enumerate}
\item 
 either (i) $\sum_{d\in V_j}d^{\rm R}\geq\frac{\epsilon}{4} w^{\rm R}_T$ for all $j={1,\ldots,h}$, 
\item or (ii) $\sum_{d\in V_j}d^{\rm I}\ge\frac{\epsilon}{4}w^{\rm I}_T $ for all $j={1,\ldots,h}$.
\end{enumerate} 
where $h\in[\frac{1}{\epsilon}-1,\frac{4}{\epsilon})$.
\end{lemma}

\begin{proof}
First, we define $\kappa \triangleq\sum_{d\in S \backslash T}d$.
If $S$ is a feasible solution to \textsc{MultiCKP$[0,\frac{\pi}{2}]$}, but $S$ is not a feasible solution to {\sc PGZ}$_T$, then at least one of the following two conditions must hold: either (i) $\kappa^{\rm R}\ge\frac{w^{\rm R}_T}{2}$ or (ii) $\kappa^{\rm I}\ge\frac{w^{\rm I}_T}{2}$ (see Fig.~\ref{f3} for an illustration). 

Without loss of generality, we assume case (i). Let us pack consecutive demands $d \in S \backslash T$ into batches such that the sum in each batch has a real component of length in $(\frac{\epsilon}{4} w^{\rm R}_T,\frac{\epsilon}{2} w^{\rm R}_T]$. More precisely, we fix an order on $S \backslash T = \{d_1, \ldots, d_r \}$, and find indices $1=k_1<k_2<\cdots <k_{h'}<k_{h'+1}=r+1$ such that
\begin{eqnarray}\label{ee1}
\sum_{k=k_\ell}^{k_{\ell+1}-1}d_k^{\rm R}&\leq&\frac{\epsilon}{2}w^{\rm R}_T,\text { for $l=1,\ldots,h'$}\\ \mbox{and \ } \sum_{k=k_\ell}^{k_{\ell+1}}d_k^{\rm R} &>&\frac{\epsilon}{2}w^{\rm R}_T\text { for $l=1,\ldots,h'-1$}. \label{ee2}
\end{eqnarray}
It follows from Eqn.~\raf{ee2} that $\sum_{k=k_\ell}^{k_{\ell+1}-1}d_k^{\rm R} >\frac{\epsilon}{4}w^{\rm R}_T $ for $l=1,\ldots,h'-1,$ since $d_{k_{\ell+1}}^{\rm R}\le\frac{\epsilon}{4}w^{\rm R}_T $. It also follows that $\frac{1}{\epsilon}\le h'<\frac{4}{\epsilon}+1$, since summing Eqn.~\raf{ee1} for $\ell=1,\ldots, h'$ yields
\begin{eqnarray}
\frac{\epsilon}{2} h' w^{\rm R}_T \ge \sum_{\ell=1}^{h'}\sum_{k=k_\ell}^{k_{\ell+1}-1}d_k^{\rm R}=\sum_{k=1}^r d_k^{\rm R}=\kappa^{\rm R}\ge \frac{w^{\rm R}_T}{2}. \label{eq:kappa}
\end{eqnarray}
The last inequality follows from our assumption.

 Similarly, summing Eqn.~\raf{ee2} for $\ell=1,\ldots, h'-1$ yields
\begin{eqnarray*}
(h'-1) \frac{\epsilon}{2} w^{\rm R}_T < \sum_{\ell=1}^{h'-1}\sum_{k=k_\ell}^{k_{\ell+1}}d_k^{\rm R}\le 2\sum_{k=1}^r d_k^{\rm R}=2\kappa^{\rm R}\le 2 w^{\rm R}_T,
\end{eqnarray*}
where the last inequality is derived by the feasibility of S.
Setting $V_\ell\triangleq\{d_{k_\ell},d_{k_\ell+1},\ldots,d_{k_{\ell+1}-1}\}$, for $\ell=1,\ldots,h'-2$,  $V_{h'-1}=\{d_{k_{h'-1}},\ldots,d_r\}$, and $h\triangleq h'-1$  satisfies the claim of the Lemma.  
\end{proof}

\subsection{Making the PTAS Truthful}\label{sec:truthful-ptas}
 In this section, we make the PTAS, presented in Sec.~\ref{sec:ptas}, truthful. 
One technical difficulty that arises in this case is that the polygons $\cP_T(\epsilon)$ defined by a guessed initial sets $T$ are not {\it monotone} w.r.t. the set of demands in $T$, that is, if we obtain $T'$ from $T$ by increasing one of the demands from $d_k$ to $d_k'\succ d_k$, then it could be the case that $\cP_T(\epsilon)\not\supseteq \cP_{T'}(\epsilon)$. Hence, it is possible to manipulate the algorithm by a selfish user in $T$ who untruthfully increases his demand to change his allocation to become a winner. To handle this issue, we will show that the number of possible polygons that arise from such a selfish user, misreporting his true demand set, and can possibly change the outcome, is only a constant in $\epsilon$ and $\delta$ (recall that we consider {\sc MultiCKP}$[0,\frac{\pi}{2}-\delta]$). Thus, it would be enough to consider only all such polygons arising from the reported demand set. Figure~\ref{fnm} provides a pictorial illustration. 
\begin{figure}
	\centering
	\includegraphics[scale=0.8]{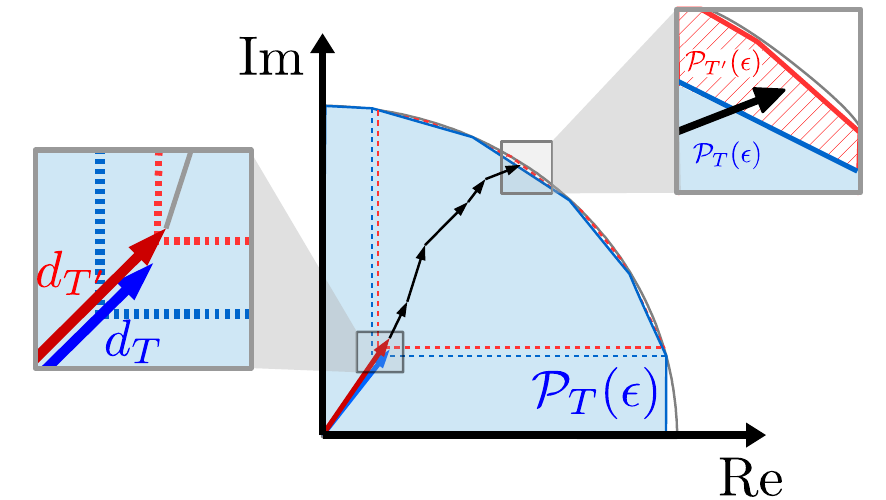}
	\caption{We consider two polygonized regions $\cR_T$ and $\cR_{T'}$, one defined by $d_T$ and the other by $d_{T'}$ (which is obtained by slightly increasing one demand in $T$). The figure illustrates the case that such a slight increase can possibly change the granularity of the grid  such that an infeasible solution in $\cP_T(\epsilon)$ would become feasible in $\cP_{T'}(\epsilon)$.     }
	\label{fnm}
\end{figure}

Since we assume that $\arg(d)\in[0,\frac{\pi}{2}-\delta]$, for all $d\in \bigcup_kD_k$, we may assume further by performing a rotation that any such vector $d$ satisfies $\arg(d)\in[\frac{\delta}{2},\frac{\pi}{2}-\frac{\delta}{2}]$. 
For convenience, we continue to denote the new demand sets by $D_k$, and redefine the valuation functions in terms of these rotated sets. By this assumption, 
\begin{equation}\label{e-a}
	\tan\frac{\delta}{2}\le\frac{d_T^{\rm I}}{d_T^{\rm R}}\le\left(\tan\frac{\delta}{2}\right)^{-1},~~\text{for any $T\subseteq\cD$}.
\end{equation}
We may also assume, by scaling $\epsilon$ by $2/(1+2\cot^2\frac{\delta}{2})$ if necessary, that
\begin{equation}\label{assum}
	\epsilon\le\frac{2}{1+2\cot^2\frac{\delta}{2}}.
\end{equation}

Now we state an important lemma that will be used to prove our main result.

\begin{lemma}\label{cl1-}
	Let $T,T'\subseteq\cD$ be such that $d_T\preceq d_{T'}$. Consider a vector $\kappa\in\CC$ such that $d_{T'}+\kappa\in \cP_{T'}(\epsilon)$. Then either (i) $d_{T}+\kappa\in\cP_{T}(\epsilon)$, or (ii) $\rho_1(T')\le \rho_1(T)+1$ and $\rho_2(T')\le \rho_2(T)+1$.
\end{lemma} 
\begin{proof}
	Suppose that $d_{T}+\kappa\not\in\cP_{T}(\epsilon)$. Since $d_{T'}\succeq d_T$, it also holds that  $d_{T'}+\kappa\not\in\cP_{T}(\epsilon)$. This implies that both $d_{T}+\kappa$ and $d_{T'}+\kappa$ lie within the same grid cell at vertical and horizontal levels $\rho_1(T)$ and  $\rho_2(T)$, respectively, because otherwise $|d_T + \kappa + (\frac{C}{2^{\rho_1(T)}} + {\bf i} \frac{C}{2^{\rho_2(T)}})|> C$ and $d_T + \kappa + (\frac{C}{2^{\rho_1(T)}} + {\bf i} \frac{C}{2^{\rho_2(T)}}) \not\in \cal{P}_{T'}(\epsilon)$ which is a contradiction. Hence $d_{T'}^{\rm R}-d_{T}^{\rm R}\le\frac{C}{2^{\rho_1(T)}}\le\frac{\epsilon w_T^{\rm R}}{4}$ and $d_{T'}^{\rm I}-d_{T}^{\rm I}\le\frac{C}{2^{\rho_2(T)}}\le\frac{\epsilon w_T^{\rm I}}{4}$. 
	
	From the definition \raf{wT} of $w_{T}^{\rm R}$, we have 
	\begin{align*}
		(w_{T}^{\rm R} + d_{T}^{\rm R})^2 = C^2 - (d_{T}^{\rm I })^2 \quad\text{ and } \quad  (w_{T'}^{\rm R} + d_{T'}^{\rm R})^2 = C^2 - (d_{T'}^{\rm I})^2
	\end{align*}
	Subtracting the above equations from each other obtains
	\begin{align}
		&(w_{T}^{\rm R} + d_{T}^{\rm R})^2 - (w_{T'}^{\rm R} + d_{T'}^{\rm R})^2 =  (d_{T'}^{\rm I })^2  - (d_{T}^{\rm I })^2 \notag\\
			\Rightarrow& (w_{T}^{\rm R} + d_{T}^{\rm R} - w_{T'}^{\rm R} - d_{T'}^{\rm R}) ({w_{T}^{\rm R} + d_{T}^{\rm R} + w_{T'}^{\rm R} + d_{T'}^{\rm R}}) = ({  d_{T'}^{\rm I} -  d_{T}^{\rm I}}) ({  d_{T'}^{\rm I} +  d_{T}^{\rm I}})\notag\\
		\Rightarrow& w_{T}^{\rm R} + d_{T}^{\rm R} - w_{T'}^{\rm R} - d_{T'}^{\rm R} =(  d_{T'}^{\rm I} -  d_{T}^{\rm I})\bigg( \frac{  d_{T'}^{\rm I} +  d_{T}^{\rm I}}{w_{T}^{\rm R} + d_{T}^{\rm R} + w_{T'}^{\rm R} + d_{T'}^{\rm R}}\bigg) \label{eq:wrl}
	\end{align}
	Therefore,
	\begin{eqnarray}\label{e1-}
	w_{T}^{\rm R}&=&w_{T'}^{\rm R}+d_{T'}^{\rm R}-d_{T}^{\rm R}+ (w_{T}^{\rm R} + d_{T}^{\rm R} - w_{T'}^{\rm R} - d_{T'}^{\rm R})\nonumber
	\\&=&w_{T'}^{\rm R}+d_{T'}^{\rm R}-d_{T}^{\rm R}+\left(d_{T'}^{\rm I}-d_{T}^{\rm I}\right)\left(\frac{d_{T'}^{\rm I}+d_{T}^{\rm I}}{d_{T'}^{\rm R}+w_{T'}^{\rm R}+d_{T}^{\rm R}+w_{T}^{\rm R}}\right)\label{eq:swrl}\\
	&\le&w_{T'}^{\rm R}+\frac{\epsilon w_T^{\rm R}}{4}+\frac{\epsilon w_T^{\rm I}}{4}\left(\frac{d_{T'}^{\rm I}+d_{T}^{\rm I}}{d_{T'}^{\rm R}+d_{T}^{\rm R}}\right)\nonumber\\
	&\le &w_{T'}^{\rm R}+\frac{\epsilon w_{T}^{\rm R}}{4}\left(1+\frac{w_T^{\rm I}}{w_T^{\rm R}}\cdot\frac{1}{\tan \frac{\delta}{2}}\right),
	\end{eqnarray}
	where we use Eqn.~\raf{eq:wrl} in Eqn.~\raf{eq:swrl}, and use Eqn.~\raf{e-a} in the last inequality. We can upper-bound $w_T^{\rm I}/w_T^{\rm R}$ by $2/\tan\frac{\delta}{2}$ also using \raf{e-a} as follows:
	\begin{align*}
			\frac{w_T^{\rm I}}{w_T^{\rm R}}=&\frac{\sqrt{1 -\left(\frac{d_T^{\rm R}}{C}\right)^2} - \frac{d_T^{\rm I}}{C}}{ \sqrt{1 -\left(\frac{d_T^{\rm I}}{C}\right)^2} - \frac{d_T^{\rm R}}{C}} =  \frac{\sqrt{1 -\left(\frac{d_T^{\rm R}}{C}\right)^2} - \frac{d_T^{\rm I}}{C}}{ \sqrt{1 -\left(\frac{d_T^{\rm I}}{C}\right)^2} - \frac{d_T^{\rm R}}{C}} \cdot 
			\frac{\bigg( \sqrt{1 -\left(\frac{d_T^{\rm R}}{C}\right)^2} + \frac{d_T^{\rm I}}{C}\bigg)\bigg( \sqrt{1 -\left(\frac{d_T^{\rm I}}{C}\right)^2} + \frac{d_T^{\rm R}}{C}\bigg)}{\bigg(\sqrt{1 -\left(\frac{d_T^{\rm R}}{C}\right)^2} + \frac{d_T^{\rm I}}{C}\bigg)\bigg( \sqrt{1 -\left(\frac{d_T^{\rm I}}{C}\right)^2} + \frac{d_T^{\rm R}}{C}\bigg)} \\
			&=\frac{\sqrt{1 -\left(\frac{d_T^{\rm I}}{C}\right)^2} + \frac{d_T^{\rm R}}{C}}{ \sqrt{1 -\left(\frac{d_T^{\rm R}}{C}\right)^2} + \frac{d_T^{\rm I}}{C}} \le\frac{1+\frac{d_T^{\rm R}}{C}}{\sqrt{1 -\left(\frac{d_T^{\rm R}}{C}\right)^2} + \frac{d_T^{\rm R}}{C}\tan\frac{\delta}{2}}
	\end{align*}
	 
	The latter quantity is bounded by $f(1)=\tfrac{2}{\tan\tfrac{\delta}{2}}$, since the function $f(a)\triangleq\frac{1+a}{\sqrt{1-a^2}+a\tan\frac{\delta}{2}}$ is monotone increasing in $a\in[0,1]$. Using this bound in \raf{e1-} and rearranging terms, we get
	\begin{equation}\label{e3-}
	w_{T'}^{\rm R}\ge w_{T}^{\rm R}\left(1-\frac{\epsilon}{4}(1+2\cot^{2}\frac{\delta}{2})\right)\ge \frac{1}{2}w_{T}^{\rm R},
	\end{equation}
	by our assumption~\raf{assum} on $\epsilon$.
	From \raf{e3-} and $\frac{\epsilon w_{T'}^{\rm R} }{8}<\frac{C}{2^{\rho_1(T')}}$, and $\frac{C}{2^{\rho_1(T)}}\le\frac{\epsilon w_{T}^{ \rm R}}{4}$, follows that $\rho_1(T')\le \rho_1(T)+1$. Similarly, we have $\rho_2(T')\le \rho_2(T)+1$. 
\end{proof}

We now state our main result for this section.
\begin{theorem}\label{t-TIE-CKP}
For any $\epsilon,\delta>0$ there is a $(1-3\epsilon)$-socially efficient truthful mechanism for {\sc MultiCKP}$[0,\frac{\pi}{2}-\delta]$. The running time is $\left|\bigcup_k D_k\right|^{O\left(\frac{\cot^2\frac{\delta}{2}}{\epsilon^2}\right)}$.
\end{theorem}
\begin{proof}
It suffices to define a declaration-independent range $\cS$ of feasible allocations, such that $\max_{\bd\in\cS}v(\bd)\ge(1-3\epsilon)\cdot\OPT$, and we can optimize over $\cS$ in the stated time.

For $T\subseteq\cD$, let $G(T)$ be the set of vectors in $\CC$ defined by the union of $\{d_T\}$ and 
\begin{itemize}
\item[(a)] the (component-wise) minimal grid points $z\in\cR_T$, such that $z=\ell_1\cap\ell_2$ for some $\ell_1\in\cL_1(\rho_1(T)+1)$ and $\ell_2\in\cL_2(\rho_2(T)+1)$, and either $\rho_1(\{z\})=\rho_1(T)+1$ {\it or} $\rho_2(\{z\})=\rho_2(T)+1$, but not both; and
\item[(b)] the (component-wise) minimal grid points $z\in\cR_T$, such that $z=\ell_1\cap\ell_2$ for some $\ell_1\in\cL_1(\rho_1(T)+1)$ and $\ell_2\in\cL_2(\rho_2(T)+1)$, and $\rho_1(\{z\})=\rho_1(T)+1$ {\it and} $\rho_2(\{z\})=\rho_2(T)+1$. 
\end{itemize}

Note that $|G(T)|=O(\frac{1}{\epsilon})$. We mention that it is possible to enumerate over all $\cL_1(\rho_1(T)+1) \cap \cL_2(\rho_2(T)+1)$ points ( i.e., the intersection points of the dotted grid lines in Figure~\ref{fz}), instead we choose to only enumerate over a smaller subset of Pareto minimal points that are defined by (a) and (b). Figure~\ref{fz} gives a pictorial example of these minimal points.

\begin{figure}
	\centering
	\includegraphics[scale=0.8]{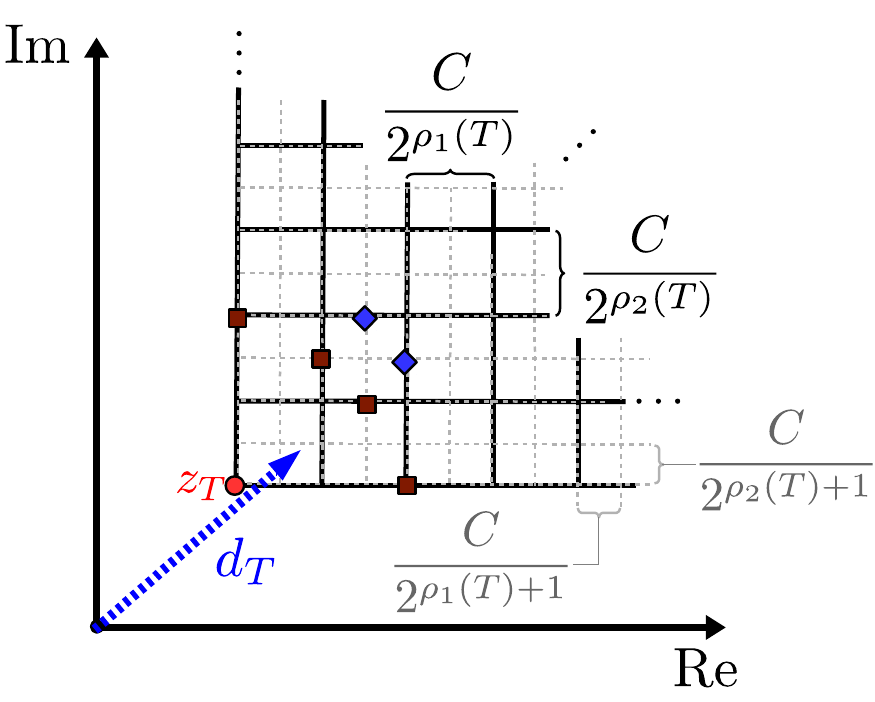}
	\caption{The figure illustrates the points of the set $G(T)$. The black solid lines are  $\cL_1(\rho_1(T) )$ and $\cL_2(\rho_2(T) )$ on the vertical and horizontal directions respectively; the gray dotted lines (with some overlapping with the black lines) correspond to $\cL_1(\rho_1(T)+1)$ and $\cL_2(\rho_2(T) +1)$. The red square points correspond to condition (a), and the blue diamond points correspond to condition (b).}
	\label{fz}
\end{figure}

For convenience of notation, let us fix two subsets $\cD_1,\cD_2\subseteq\cD$.
For $z\in G(T)$, let us denote by $\cS_z(\cD_2)$ the range of feasible allocations defined as in \raf{range} with respect to the \textsc{Multi-$m$DKP} problem with constraints \raf{mCm-1}-\raf{mCm-2}, when 
\begin{itemize}
\item[(I)] $T$ is replaced by $T\cup\{z-d_T\}$ (and hence, $z$ is used to define the polygon $\cP_T(\epsilon)$);
\item[(II)] we add an additional ``dummy'' user $n+1$ to $\cN$ with valuation $v_{n+1}(d)=0$ for all $d\in\cD$, such that the vector $z-d_T$ as allocated to this user; and 
\item[(III)] the set of vectors in $\cN\backslash N$ is chosen from $\cD_2$. 
\end{itemize}
Then we define the range $\cS(\cD_1,\cD_2)$ as the union: $$\cS(\cD_1,\cD_2)\triangleq\bigcup_{T\subseteq\cD_1:~|T|\le\frac{1}{\epsilon}}\left(\bigcup_{z\in G(T)}\cS_{z}(\cD_2)\right).$$ 
By Lemmas \ref{l1-} and \ref{main-lem}, we have $\max_{\bd\in\cS(\cD,\cD)}v(\bd)\ge(1-3\epsilon)\OPT$ (since $d_T\in G(T)$). It remains to argue that we can efficiently  optimize over $\cS(\cD,\cD)$. 
Using Lemma~\ref{cl1-}, we argue that we can solve the optimization problem over $\cS(\cD,\cD)$ assuming that $\cD=\bigcup_{k}D_k$, that is, $\max_{\bd\in \cS(\cD,\cD)}v(\bd)=\max_{\bd\in \cS(\bigcup_{k}D_k,\bigcup_{k}D_k)}v(\bd).$ 
One direction ``$\geq$'' is obvious;
so let us show that $\max_{\bd\in \cS(\cD,\cD)}v(\bd)\leq\max_{\bd\in \cS(\bigcup_{k}D_k,\bigcup_{k}D_k)}v(\bd).$ 

Suppose that $\bd^*=(d_1^*,\ldots,d_n^*)$ is an optimal allocation over $\cS(\cD,\cD)$, but such that $\bd^*\in\cS_{z'}$ for some $z'\in G(T'')$,  $T''\subseteq\cD$, and $T''\not\subseteq \bigcup_kD_k$. Then let us show that there is a set $T\subseteq \bigcup_kD_k$, $z\in G(T)$, and $\widetilde\bd\in\cS_z(\cD)$, such that  $v(\widetilde\bd)=v(\bd^*)$.

Define an allocation $\widetilde\bd$ as follows: Let $N=\{k:~d_k^*\in T''\}$; for each $k\in N$, we choose $\widetilde d_k\in D_k$ such that $\widetilde d_k\preceq d_k^*$ and $v_k(\widetilde d_k)=v_k(d_k^*)$, and we keep $\widetilde d_k=d_k^*$ if $k\not\in N$.
Let us apply the statement of the lemma with $T=\{\widetilde d_k:~k\in N\}$, $T'=T''\cup\{z'-d_{T''}\}$, and $\kappa=\sum_{k:k\not\in N\cup\{n+1\}}d_k^*$. If (i) holds then
$d_T+\kappa\in \cP_T(\epsilon)$ and therefore we have  
\begin{equation}\label{oneDir}
\max_{\bd\in\cS(\cD,\cD)}v(\bd)=\max_{\bd\in \cS(\bigcup_kD_k,\cD)}v(\bd).
\end{equation}
On the other hand,
if (ii) holds, then $\rho_1(T')\in\{\rho_1(T),\rho_1(T)+1\}$ and $\rho_2(T')\in\{\rho_2(T),\rho_2(T)+1\}$. In this case, if $\rho_1(T')=\rho_1(T)$ and $\rho_2(T')=\rho_2(T)$ then $\cP_{T'}(\epsilon)\subseteq     
\cP_{T}(\epsilon)$ (since $d_T\preceq d_{T'}$), in contradiction that (i) does not hold; otherwise, there is a point $z\in G(T)$ such that $z\preceq z'$, $\rho_1(T\cup\{z-d_T\})=\rho_1(T')$ and $\rho_2(T\cup\{z-d_T\})=\rho_2(T')$. Then $z+\kappa\preceq z'+\kappa\in\cP_{T'}(\epsilon)\subseteq\cP_{T\cup\{z-d_T\}}(\epsilon)$, and we get again \raf{oneDir}. 

Finally, we note that
\begin{equation*}
\max_{\bd\in \cS(\bigcup_kD_k,\cD)}v(\bd)=\max_{\bd\in\cS(\bigcup_{k}D_k,\bigcup_kD_k)}v(\bd),
 \end{equation*}as follows from (the proof of) Lemma~\ref{l2-}.

\end{proof}

\section{A Truthful $(1, 1+\epsilon)$-FPTAS for {\sc MultiCKP}$[0,\pi - \delta]$}\label{sec:tbp}
 In this section, we present our truthful $(1,1+\epsilon)$-approximation for {\sc MultiCKP}$[0, \pi-\delta]$ by a reduction to {\sc Multi-$2$DKP}.
As in \cite{KTV13}, the basic idea is to round off the set of possible demands to obtain a range, by which we can optimize over in polynomial time using dynamic programming (to obtain an MIR). 

Let $\theta = \max\{\phi - \frac{\pi}{2},0\}$, where $\phi\triangleq\max_{d\in\cD}{\rm arg}(d)$. We assume that $\tan \theta$ is bounded by an {\it a-priori} known polynomial $P(n)\ge 1$ in $n$, that is {\it independent} of the customers declarations (valuations and demands), because the power factors are often limited by certain regulations in practice.

Let $\cN_+\triangleq \{k \in \cN \mid d^{\rm R} \ge 0~\forall d\in D_k\}$ and $\cN_-\triangleq\{k \in \cN \mid d^{\rm R} < 0~\forall d\in D_k\}$ be the subsets of users with demand sets in the first and second quadrants respectively (recall that we restrict users' declarations to allow such a partition).

We can upper bound the total projections for any feasible allocation $\bd=(d_1,\ldots,d_n)$ of demands as follows:
\begin{align}
 \sum_{k \in \cN} d_k^{\rm I} \le C, \quad
 \sum_{k \in \cN_-(\bd) } - d_k^{\rm R}   \le  C \tan \theta, \quad 
 \sum_{k \in \cN_+(\bd)}  d_k^{\rm R}   &\le C(1+ \tan \theta), \label{eq:ubounds}
\end{align}
where $\cN_+(\bd) \triangleq \{ k\in \cN \mid d_k^{\rm R} \ge 0\}\subseteq\cN_+$ and $\cN_-(\bd) \triangleq \{ k \in \cN\mid   d_k^{\rm R} < 0 \}\subseteq\cN_-$.
Define $L \triangleq \frac{\epsilon C}{n (P(n)+1)}$, and for $d\in\cD$, define the new rounded demand $\widehat{d}$ as follows:
\begin{equation}
\widehat d =
\widehat d^{\rm R} + {\bf i} \widehat d^{\rm I} \triangleq 
\left\{\begin{array}{ll}
\left\lceil \frac{d^{\rm R}}{L} \right\rceil \cdot L + {\bf i} \left\lceil \frac{d^{\rm I}}{L} \right\rceil \cdot L, &\text{ if }d^{\rm R}\ge 0,\\[3mm]
\left\lfloor \frac{d^{\rm R}}{L} \right\rfloor \cdot L + {\bf i} \left\lceil \frac{d^{\rm I}}{L} \right\rceil \cdot L, & \text{ otherwise. } 
\end{array}\right.
\label{eq:truc}
\end{equation}
For convenience, we will write $\widehat\bd=(\widehat d_1,\ldots,\widehat d_n)$. Note that by this definition, $\cN_+(\bd)=\cN_+(\widehat\bd)$ and $\cN_-(\bd)=\cN_-(\widehat\bd)$. 
 
Consider an optimal allocation $\bd^\ast=(d_1^\ast,\ldots,d_n^\ast)$ to \textsc{MultiCKP} $[0,\pi-\delta]$.
Let $\xi_+$ (and $\xi_-$), $\zeta_+$ (and $\zeta_-$) be the respective real and imaginary absolute total projections of the rounded demands in $\cN_+(\bd^\ast)$ (and $\cN_-(\bd^\ast)$).
Then the possible values of $\xi_+, \xi_-, \zeta_+, \zeta_-$ are integral mutiples of $L$ in the following ranges:
\begin{align*}
 \xi_+ \in {\cal A}_+ & \triangleq \left\{0, L, 2L,\ldots,\left\lceil \frac{C (1 + P(n) )}{L} \right\rceil \cdot L\right\},\\
\xi_- \in {\cal A}_-& \triangleq \left\{0,L, 2L,\ldots, \left\lceil \frac{C \cdot P(n) }{L} \right\rceil\cdot L \right\},\\
\zeta_+,\zeta_-  \in {\cal B}&  \triangleq \left\{0, L, 2L,\ldots,\left\lceil \frac{C}{L} \right\rceil\cdot L\right\}.
\label{eq:grid}
\end{align*}

We first present a  $(1,1+4\epsilon)$-approximation algorithm ({\sc MultiCKP-FPTAS}) for \textsc{MultiCKP$[0,\pi-\delta]$}; then we show how to implement it as an MIR mechanism.

The basic idea of Algorithm {\sc MultiCKP-FPTAS} is to enumerate the guessed total projections on real and imaginary axes for $\cN_+(\bd^\ast)$ and $\cN_-(\bd^\ast)$ respectively. 
We then solve two separate {\sc Multi-2DKP} problems (one for each quadrant) to find subsets of demands that satisfy the individual guessed total projections. But since {\sc Multi-2DKP} is generally NP-hard, we need first to round the demands to get a problem that can be solved efficiently by dynamic programming. We note that the violation of the optimal solution to the rounded problem w.r.t. the original problem is small in $\epsilon$. 

For any allocation $\bd=(d_1,\ldots,d_n)\in\cD^n$, let us for brevity write $\tau_+(\bd)\triangleq\sum_{k\in\cN_+(\bd)}d_k^{\rm R}$, $\tau_-(\bd)\triangleq\sum_{k\in \cN_-(\bd)}-d_k^{\rm R}$, and $\tau_I(\bd)\triangleq\sum_{k\in \cN}d_k^{\rm I}$. Then by \raf{eq:truc} and the fact that $x \le t \lceil \frac{x}{t} \rceil  \le x + t$ for any $x,t$ such that $t>0$, we have 
\begin{eqnarray}\label{eq:bds}
&\max\{\tau(\widehat \bd)-nL,0\}\le \tau(\bd)\le \tau( \widehat\bd),
\end{eqnarray}
for all $\tau\in\{\tau_+,\tau_-,\tau_I\}$. 
\begin{lemma}
For any feasible allocation $\bd=(d_1,\ldots,d_n)$ to \textsc{MultiCKP $[0,\pi-\delta]$}, we have $\big | \sum_{k} \widehat d_k\big| \le  ( 1 + 2\epsilon)C$.
 \label{lem-trunc}
\end{lemma}
\begin{proof}
Using~\raf{eq:bds} and~\raf{eq:ubounds},
\begin{align*}
\left( \sum_{k\in\cN} \widehat d_k^{\rm R} \right)^2 +  \left(\sum_{k \in \cN} \widehat d_k^{\rm I} \right)^2 &=\left(\tau_+(\widehat \bd)-\tau_-( \widehat \bd)\right)^2 +   \tau_I^2( \widehat \bd)\nonumber\\
&=\tau_+^2( \widehat \bd)+\tau_-^2(\widehat\bd)-2\tau_+(\widehat \bd) \tau_-(\widehat\bd)+\tau_I^2(\widehat\bd)\nonumber\\ 
&\le (\tau_+(\bd)+nL)^2+(\tau_-(\bd)+nL)^2-2\tau_+(\bd) \tau_-(\bd)+(\tau_I(\bd)+nL)^2\nonumber\\
&=(\tau_+(\bd)-\tau_-(\bd))^2+\tau_I^2(\bd)+2nL (\tau_+(\bd)+\tau_-(\bd)+\tau_I(\bd))+3n^2L^2\nonumber\\
&= \left(\sum_{k\in \cN} d_k^{\rm R} \right)^2 + \left(\sum_{k\in \cN} d_k^{\rm I}  \right)^2 +2nL\left(\sum_{k\in \cN}|d_k^{\rm R}|+\sum_{k\in\cN}d_k^{\rm I}\right)+3n^2L^2\nonumber\\
&\le C^2 + 4nL (P(n)+1) C + 3n^2L^2 = C^2 + 4 \epsilon C^2 + 3\epsilon^2 C^2/(1+P(n))^2\nonumber \\ &\le C^2 (1+4\epsilon + 3\epsilon^2)\le C^2(1+2\epsilon)^2. 
\end{align*}
\end{proof}

The next step is to solve the rounded instances exactly. This can be done with essentially the same dynamic program used in the proof of Lemma~\ref{l2-}, modulo a small modification, which we include here for completeness. 

For $d\in\cD$, we use $\bar d$ to denote the vector in $\CC$ with components $\bar d^{\rm R}=-d^{\rm R}$ and $\bar d^{\rm I}=d^{\rm I}$.
We define a (new) valuation function $\bar v_k$ by: $\bar v_k(d)=v_k(d)$, for $k\in\cN_+$, and $\bar v_k(d)=v_k(\bar d)$, for $k\in\cN_-$.  Let further $\widehat{\cD}\triangleq \{ \frac{d}{L} \in \cD :~ d^{\rm R}\in\cA_+ \text{ and }d^{\rm I}\in\cB\},$ and note that $|\widehat\cD|=O(\frac{n^2P^3(n)}{\epsilon^2})$.

Assume an arbitrary order on $\cN = \{ 1, ..., n\}$. We define a 3D table, with each entry ${U}(k,c)$ being the maximum value obtained from a subset of users $\{1,2,\dots,k\} \subseteq \cN$, each choosing a demand from $\widehat{\cD}$, such that the chosen demands fit exactly  within capacity $c\in\widehat\cD$ (i.e., satisfy the capacity constraints as an equation along each of the axes). 
The cells of the table are defined according to the following rules:

\begin{align*}
{U}(1,c) &\triangleq  \bar v_1(L\cdot c);
\\
U(k,  c)& \triangleq -\infty\text{ for all } c\not \in\widehat\cD;\\
U(k+1,c) &\triangleq  \max_{c' \in \widehat \cD} \{\bar v_{k+1}(L\cdot c') + {U}(k, c - c')\}.
\end{align*}
The corresponding optimal allocation ($F_+$ or $F_-$) can be mapped to the original range of demands $\cD$ as follows:
\begin{align*}
\cI(1,c) &\triangleq \{(1,d_1) \}\text{ where $d_1\in D_1$ is}\\ 
&\text{  s.t. } \left\{
\begin{array}{ll}
v_1(d_1)=v_1(L\cdot c)\text{ and }d_1\preceq L\cdot c,&\text{ if $1\in\cN_+$,}\\ 
v_1(d_1)=v_1(L\cdot \bar c)\text{ and }d_1\preceq L\cdot \bar c,&\text{ if $1\in\cN_-$;} 
\end{array}
\right.
\\
{\cI}(k+1,c) &\triangleq  \cI(k, c) \cup \{(k + 1, d_{k+1} )\} \text{ where $d_{k+1}\in D_{k+1}$ is}\\ &\text{ s.t. }\left\{
\begin{array}{ll}
v_{k+1}(d_{k+1})=v_{k+1}(L\cdot d)\text{ and }d_{k+1}\preceq L\cdot d,&\text{ if $k+1\in\cN_+$,}\\ 
v_{k+1}(d_{k+1})=v_{k+1}(L\cdot \bar{d})\text{ and }d_{k+1}\preceq L\cdot \bar{d},&\text{ if $k+1\in\cN_-$,} 
\end{array}
\right.  
\\
&\text{where }  d\in\argmax _{c' \in \widehat \cD} \{\bar v_{k+1}(L\cdot c') + {U}(k, c - c')\}.
\end{align*}
This table can be filled-up by standard dynamic programming; we denote such a program by {\sc Multi-2DKP-Exact}$[\cdot]$. 

\begin{algorithm}[!htb]
\caption{{\sc MultiCKP-FPTAS} $( \{v_k,D_k\}_{k \in \cN}, C,\epsilon) $}\label{MultiCKP-biFPTAS}
\begin{algorithmic}[1]
\Require Users' multi-minded valuations $\{v_k,D_k\}_{k\in \cN}$; capacity $C$; accuracy parameter $\epsilon$
\Ensure a $(1,1+4\epsilon)$-approximation $\bd=({d}_1,\ldots,d_n)$ to \textsc{MultiCKP$[0,\pi-\delta]$}
\State $({d}_1,\ldots,d_n) \leftarrow (\bzero,\ldots,\bzero)$
\State $\widehat{\cD}\leftarrow \{ \frac{d}{L} \in \cD :~ d^{\rm R}\in\cA_+ \text{ and }d^{\rm I}\in\cB\}$ 
\ForAll {$\xi_+ \in {\cal A}_+, \xi_- \in {\cal A}_-, \zeta_+, \zeta_- \in {\cal B}$}
\If {$(\xi_+ - \xi_-)^2 + (\zeta_+ + \zeta_-)^2 \le (1+2\epsilon)^2C^2$}\label{cond1}
\State {\small $F_+ \leftarrow \text{\sc Multi-2DKP-Exact}(\{v_k,D_k\}_{k\in\cN_+}, \frac{\xi_+}{L},\frac{\zeta_+}{L},\widehat\cD)$}
\State {\small $F_- \leftarrow \text{\sc Multi-2DKP-Exact}(\{\bar v_k,D_k\}_{k\in \cN_-}, \frac{\xi_-}{L},\frac{\zeta_-}{L},\widehat\cD)$} 
\State $(d_1',\ldots,d_n')\leftarrow F_+ \cup F_-$ 
\If{$\sum_kv_k(d'_k) > \sum_kv_k(d_k)$}
\State $(d_1,\ldots,d_n)\leftarrow  (d_1',\ldots,d_n')$
\EndIf 
\EndIf
\EndFor
\State \Return $(d_1,\ldots,d_n)$
\end{algorithmic}
\end{algorithm}

The following lemma states that the allocation returned by {\sc MultiCKP-FPTAS} does not violate the capacity constraint by more than a factor of $1+4\epsilon$.
\begin{lemma}\label{lem-trunc2}
Let $\bd$ be the allocation returned by {\sc MultiCKP-FPTAS}. Then $|\sum_{k} d_k|\le(1+4\epsilon) C$. 
\end{lemma}
\begin{proof}
As in the proof of Lemma~\ref{lem-trunc},
\begin{eqnarray}
\label{eq:ee1}
\left( \sum_{k \in \cN} d_k^{\rm R} \right)^2 +  \left(\sum_{k \in \cN} d_k^{\rm I} \right)^2 &=&\left(\tau_+(\bd)- \tau_-(\bd)\right)^2 +  \tau_I^2(\bd)\nonumber\\
&=&\tau_+^2(\bd)+\tau_-^2(\bd)-2\tau_+(\bd)\tau_-(\bd)+\tau_I^2(\bd).
\end{eqnarray} 
If both $\tau_+(\widehat\bd)$ and $\tau_-(\widehat\bd)$  are less than $nL$, then the R.H.S. of \raf{eq:ee1} can be bounded by 
\begin{eqnarray}\label{eq:ee2}
\tau_+^2(\widehat\bd)+\tau_-^2(\widehat\bd)+\tau_I^2(\widehat\bd)
&\le&\tau_+^2(\widehat\bd)+\tau_-^2(\widehat\bd)-2\tau_+(\widehat\bd) \tau_-(\widehat\bd)+2n^2L^2+\tau_I^2(\widehat\bd)\nonumber\\
&=&(\tau_+(\widehat\bd)- \tau_-(\widehat\bd))^2+\tau_I^2(\widehat\bd)+2n^2L^2.
\end{eqnarray}
Otherwise, we bound the R.H.S. of \raf{eq:ee1} by
\begin{eqnarray}\label{eq:ee3}
&\tau_+^2(\widehat\bd)+\tau_-^2(\widehat\bd)-2(\tau_+(\widehat\bd)-nL)(\tau_-(\widehat\bd)-nL)+\tau_I^2(\widehat\bd)\nonumber\\
&=(\tau_+(\widehat\bd)-\tau_-(\widehat\bd))^2+\tau_I^2(\widehat\bd)+2nL(\tau_+(\widehat\bd)+ \tau_-(\widehat\bd))-2n^2L^2.
\end{eqnarray}
Since $\bd=F_+\cup F_-$ is obtained from feasible solutions $F_+$ and $F_-$ to $\text{\sc 2DKP-Exact}(\{v_k,D_k\}_{k\in\cN_+}, \frac{\xi_+}{L},\frac{\zeta_+}{L},\widehat\cD)$
 and $\text{\sc 2DKP-Exact}(\{\bar v_k,D_k\}_{k\in \cN_-}, \frac{\xi_-}{L},\frac{\zeta_-}{L},\widehat\cD)$, respectively, and  $\xi_+ , \xi_-,\zeta_+, \zeta_-$ satisfy the condition in Step~\ref{cond1}, it follows from \raf{eq:ee1}-\raf{eq:ee3} that
\begin{align*}
\left( \sum_{k \in \cN} d_k^{\rm R} \right)^2 +  \left(\sum_{k \in \cN} d_k^{\rm I} \right)^2 &\le \left( \sum_{k \in \cN} \widehat d_k^{\rm R}\right)^2 +  \left(\sum_{k \in \cN}\widehat d_k^{\rm I} \right)^2+2nL\sum_{k \in \cN} |\widehat d_k^{\rm R}|+2n^2L^2\nonumber\\
&=(\xi_+ - \xi_-)^2 + (\zeta_+ + \zeta_-)^2 + 2nL(\xi_++\xi_-)+2n^2L^2\nonumber\\
&\le \left((1+2\epsilon)^2C^2 +4n\frac{\epsilon C}{n(P(n)+1)} (1+P(n))C + 2n^2\frac{\epsilon^2C^2}{n^2(P(n)+1)^2}\right)\\
&\le \left((1+2\epsilon)^2 +4\epsilon + 2\epsilon^2\right)C^2\le (1+4\epsilon)^2C^2.
\end{align*}
\end{proof}

\begin{theorem}
For any $\epsilon,\delta>0$, there is a truthful mechanism for {\sc MultiCKP$[0,\pi - \delta]$}, that returns a $(1,1+4\epsilon)$-approximation. The running time is polynomial in $n$, $\cot\delta$,  and $\frac{1}{\epsilon}$.
\end{theorem}
\begin{proof}
We define a declaration-independent range $\cS$ as follows. 
For $\xi_+ \in {\cal A}_+, \xi_- \in {\cal A}_-, \zeta_+, \zeta_- \in {\cal B}$, define
\begin{align*}
\cS_{\xi_+,\xi_-,\zeta_+,\zeta_-}&\triangleq\{\bd=(d_1,\ldots,d_n)\in \widehat\cD^n:~\sum_{k\in\cN_+}L\cdot d_k^{\rm R}=\xi_+,~\sum_{k\in\cN_+}L\cdot d_k^{\rm I}=\zeta_+,\\
&~~~~-\sum_{k\in\cN_-}L\cdot d_k^{\rm R}=\xi_-,~\sum_{k\in\cN_-}L\cdot d_k^{\rm I}=\zeta_-\}.
\end{align*}
Define further 
\begin{align*}
\cS\triangleq\bigcup_{(\xi_+ - \xi_-)^2 + (\zeta_+ + \zeta_-)^2 \le (1+2\epsilon)^2C^2}\cS_{\xi_+,\xi_-,\zeta_+,\zeta_-}.
\end{align*}
Using Algorithm {\sc MultiCKP-FPTAS}, we can optimize over $\cS$ in time polynomial in $n$ and $\frac{1}{\epsilon}$.  
Thus, it remains only to argue that the algorithm returns a $(1,1+4\epsilon)$-approximation w.r.t. the original range $\cD^n$. 
To see this, let $d_1^*,\ldots,d_n^*\in\cD$ be the demands allocated in an optimum solution to {\sc MultiCKP$[0,\pi - \delta]$}, and $d_1,\ldots,d_n\in\cD$ be the demands allocated by {\sc MultiCKP-FPTAS}. Then by Lemma~\ref{lem-trunc}, the truncated optimal allocation $(\widehat d_1^*,\ldots,\widehat d_n^*)$ is feasible with respect to a capacity of $(1+2\epsilon)C$, and thus its projections will satisfy the condition in Step~\ref{cond1} of Algorithm~\ref{MultiCKP-biFPTAS}. It follows that $v(\bd) \ge v(\widehat\bd^*)\ge v(\bd^*)=\OPT$, where the second inequality follows from the way we round demands~\raf{eq:truc} and the monotonicity of the valuations. Finally, the fact that the solution returned by{\sc MultiCKP-FPTAS} violates the capacity constraint by a factor of at most $(1+4\epsilon)$ follows readily from Lemma~\ref{lem-trunc2}. 
\end{proof}

\section{Conclusion}

In this paper, we provided truthful mechanisms for an important variant of the knapsack problem with complex-valued demands. We gave a truthful PTAS when all demand sets of users lie in the positive quadrant (which is attained by limited power factors), and a truthful FPTAS with capacity augmentation when some of the demand sets can lie in the second quadrant (which captures the general setting of large power factors). Our hardness results show that this is essentially best possible assuming P$\neq$NP. Our fundamental results underpin a wide class of resource allocation problems arising in smart grid. The complete understanding of truthful complex-demand problem paves the way to more sophisticated and efficient mechanism design in future smart grids.  Recently, this work has been extended to consider scheduling problems \cite{KKCEZ16,KKEC16CSP,KKCE16b}.


\bibliographystyle{plain}
\bibliography{reference}
\appendix
\section*{APPENDIX}

\section{Proofs of Lemmas~\ref{l1-} and~\ref{l2-}}

\noindent{\bf Lemma~\ref{l1-}.}  
\begin{proof}
	
Consider an optimal allocation $\bd=(d_1,\ldots,d_n)$, and assume without loss of generality that $v_1(d_1)\ge \ldots\ge v_n(d_n)$, and (by the monotonicity of valuations) that $\sum_{k\in \cN}d_k=c$. Let $t\triangleq \lceil\frac{m}{\epsilon}\rceil$, and $N=\{1,2,\ldots,t\}$. Without loss of generality, suppose $c^i-\sum_{k\in N}d_k^i > 0$ for all $i$ and $n>t$. Then, there exist $k_1,\ldots,k_m\in\cN\setminus N$ such that, for all $i\in[m]$, $d_{k_i}^i\ge\frac{c^i-\sum_{k\in N}d_k^i}{n-t}$. We define another allocation $\widehat\bd=(\widehat d_1,\ldots,\widehat d_n)\in\cD^n$ as follows: let $\bar\bd=(d_k:~k\in N)$, and for $i\in[m]$, set
\begin{equation*}\label{round-}
\widehat d_k^i=\left\{
\begin{array}{ll}
d_k^i&\text{ if }k\in N,\\
0&\text{ if }k=k_i,\\
\left\lceil \frac{d_k^i}{b^i_{N,\bar \bd}}\right\rceil\cdot b_{N,\bar\bd}^i&\text{ if }k\not\in N \text{ and }k\neq k_i.
\end{array}
\right.
\end{equation*} 
Note that
\begin{align}
	\sum_{k\not\in N}\widehat d_k^i\leq\sum_{k\not\in N,~k\ne k_i}(d_k^i+b_{N,\bar\bd}^i)\le\sum_{k\not\in N}d_k^i+(n-t) b_{N,\bar\bd}^i-d_{k_i}^i\le c^i-\sum_{k\in N}d_k^i \label{eq:apx51}
\end{align}
 that is, $\widehat\bd$ is a feasible allocation.
Let $r_k^i\triangleq\left\lceil \frac{d_k^i}{b^i_{N,\bar\bd}}\right\rceil$ for $k\in\cN\backslash(N\cup\{k_i\})$ and $r_k^i\triangleq 0$ for $k=k_i$. Then it follows using Eqn.~\raf{eq:apx51} that $\sum_{k\not\in N}r_k^i\le (n-t)^2$:
\begin{align}
	\sum_{k\not\in N}r_k^i \le \sum_{k\not\in N,~k\ne k_i} \frac{d_k^i + b^i_{N,\bar \bd} }{b^i_{N,\bar \bd}} = (n-t)^2  \frac{\sum_{k\not\in N,~k\ne k_i} (d_k^i + b^i_{N,\bar \bd}) }{c^i-\sum_{k\in N}d_k^i} \le (n-t)^2
\end{align}
Thus $\widehat\bd\in\cS$. Finally note that, for all $i\in[m]$, $v_k(d_{k_i})\le\frac{1}{t}\sum_{k\in N}v_k(d_k)$, from which we get by the monotonicity of the valuations that $v(\widehat\bd)=\sum_{k\in N}v_k(d_k)+\sum_{k\not\in N,~k\not\in\{ k_1,\ldots,k_m\}}v_k(\widehat d_k)\ge\sum_{k}v_k(d_k)-\sum_{i=1}^mv_k(d_{k_i})\ge(1-\epsilon)\OPT.$
\end{proof}

\noindent{\bf Lemma~\ref{l2-}.}
\begin{proof}
We first observe that, due to the way the valuations are defined in \raf{mm-val}, we may assume for  the purpose of computing an optimal allocation $\bd^*$ that $\cD=\bigcup_{k}D_k$. Indeed, suppose that $\bd^*=(d_1^*,\ldots,d_n^*)\in \cS_{N,\bar\bd^*}$, where $\bar\bd^*=(d_k^*:~k\in N)$, $d^*_{k'}\not\in D_{k'}$ for some $k'\in N$, and  $d^*_k=(r_k^i\cdot b_{N,\bd^*}^i:~i\in[m])$ for $k \not\in N$. Then let us define a new allocation $\widetilde d$ as follows: for each $k\in N$, we choose $\widetilde d_k\in D_k$ such that $\widetilde d_k\preceq d_k^*$ and $v_k(\widetilde d_k)=v_k(d_k^*)$; we set $\bar\bd=(\widetilde d_k:~k\in N)$, and for $k\not\in N$, define $\widetilde d_k=(r^i_k \cdot b^i_{N,\bar\bd}:~i\in[m])$. Note by \raf{bdT} that $b_{N,\bar\bd}\ge b_{N,\bar\bd^*}$, and hence $v(\widetilde \bd)\ge v(\bd^*)$. We note furthermore that $\widetilde\bd\in\cS_{N,\bar \bd}$, since for all $i$, we have 
{\small
\begin{align*}
\sum_k (\widetilde d_k^i-d_k^{*,i})&=\sum_{k\in N}(\widetilde d_k^i-d_k^{*,i})+\sum_{k'\not\in N}\frac{r^i_{k'}}{(n-{|N|})^2}\sum_{k\in N}(d_k^{*,i}-\widetilde d_k^i)\\
&=\sum_{k\in N}(\widetilde d_k^i-d_k^{*,i})\left(1-\frac{\sum_{k'\not\in N}r^i_{k'}}{(n-{|N|})^2}\right)\le 0,
\end{align*}}
\hspace{-0.05in}since $\widetilde d_k^i\le d_k^{*,i}$ and $\sum_{k'\not\in N}r^i_{k'}\le(n-{|N|})^2$, for all $i$. It follows that $\sum_k\widetilde d_k\le\sum_k{d_k}^*\le c$, and hence $\widetilde d\in\cS_{N,\bar \bd}$ as claimed.

To maximize over $\cS$, with the restriction that $\cD=\bigcup_{k}D_k$, we iterate over all subsets $N\subseteq \cN$ of size at most $\frac{m}{\epsilon}$ and all partial selections $\bar\bd=(d_k\in D_k:~k\in N)$. For each such choice $(N,\bar\bd)$, we use dynamic programming to find $\argmax_{\bd\in\cS_{N,\bar\bd}}v(\bd)$. Let $b_{N,\bar \bd}$ be as defined in \raf{bdT}. Without loss of generality, assume $\cN\setminus N = \{ 1, \ldots, n-t\}$. For $k\in\cN\setminus N$ and $r=(r^1,\ldots,r^m)\in\{0,1,\ldots,(n-|N|)^2\}^m$, define ${U}(k,r)$ to be the maximum value obtained from a subset of users $\{1,2,\dots,k\} \subseteq \cN\setminus N$, with user $j\in[k]$ having demand $\widehat d_j^i=r_j^i\cdot b_{N,\bar\bd}^i$ for  $i\in[m]$, where $r_j^i\in\{0,1,\ldots,(n-|N|)^2\}$, and such that $\sum_{j\in[k]}r_j^i\le r^i$. For two vectors $x,y\in\RR^m$, let us denote by $x*y$ the vector with components $(x_1y_1,\ldots,x_my_m)$. Define ${U}(1,r)=-\infty$, if $r\not\geq \bzero$. Then we can use the following recurrence to compute $U(k,r)$:
\begin{eqnarray*}
{U}(1,r) &=&  \max_{r}v_1(b_{N,\bar d}*r)\\
{U}(k+1,r) &=& \max_{r_{k+1}\le r}\left\{v_{k+1}(b_{N,\bar d}*r_{k+1})+U(k,r-r_{k+1})\right\}.
\end{eqnarray*}
Note that the number of possible choices for $r$ is at most $n^{2m}$, and hence the total time required by the dynamic program is $n^{O(m)}$. 
Finally, given the vector $r$ that maximizes $U(n-|N|,r)$, we can obtain (by tracing back the optimal choices in the table) an optimal vector $r_k=(r_k^1,\ldots,r_k^m)$, for each $k\in\cN\setminus N$.
From this, we get an allocation $\widetilde\bd\in\cS$, by defining $\widetilde d_k=d_k$ for $k\in N$ and, for $k\not\in N$, we choose $\widetilde d_k\in D_k$ such that $\widetilde d_k\preceq r_k*b_{N,\bar\bd}$ and $v_k(\widetilde d_k)=v_k(r_k*b_{N,\bar\bd})$. 
\end{proof}

\end{document}